\documentclass[manuscript]{acmart} 

\setcopyright{none}

\usepackage{smartdiagram}
\usepackage{tikz}
\usetikzlibrary{automata, positioning, calc, decorations, shapes, fit, shadows, arrows}

\usepackage{mathtools}
\usepackage{cancel}

\usepackage{amsmath}
\usepackage{amssymb,latexsym}
\usepackage{amsmath}
\usepackage{verbatim}
\usepackage{graphicx}
\usepackage{syntax}
\usepackage{breqn}

\usepackage{wrapfig}
\usepackage{lipsum}

\usepackage{algorithm}
\usepackage{algpseudocode}

\usepackage{xcolor}
\usepackage{stmaryrd}
\usepackage{paralist}
\usepackage{braket}

\usepackage{subfig}

\usepackage{mathpartir}

\usepackage{esvect}

\usepackage{url}

\usepackage{times}
\usepackage{microtype}

\usepackage{xspace}

\newcommand{\ppf}{\ensuremath{\mathcal{PPF}}}

\newcommand{\ppl}{\ensuremath{\mathcal{PPL}}}
\newcommand{\kbl}{\ensuremath{\mathcal{KBL}}}
\newcommand{\sat}{\ensuremath{\models}}
\newcommand{\der}{\ensuremath{\vdash}}
\newcommand{\conf}{\ensuremath{\models_C}}

\newcommand{\poval}{\ensuremath{\pi}}

\newcommand{\poSet}{\ensuremath{\Pi}}
\newcommand{\conSet}{\ensuremath{\mathcal{C}}}
\newcommand{\struct}{\ensuremath{\mathcal{A}}}

\newcommand{\rtpolicy}[3]{\ensuremath{\llbracket #1 \rrbracket_{#2}^{#3}}}
\newcommand{\kb}{\ensuremath{\mathit{KB}}}

\newcommand{\evtSet}{\ensuremath{\mathit{EVT}}}

\newcommand{\wkbl}{\ensuremath{\mathcal{F_{KBL}}}}

\newcommand{\actSet}{\ensuremath{\Sigma}}
\newcommand{\etal}{\mbox{\emph{et al.\ }}}

\newcommand{\timestampset}{\ensuremath{\mathbb{T}}}

\usepackage{color}
\definecolor{dark-green}{rgb}{0.0, 0.5, 0.0}

\newcommand{\nat}{\ensuremath{\mathbb{N}}}

\newcommand{\mrstn}{\ensuremath{\mathcal{M}^{\mathit{rst}}_n}}

\newcommand{\langn}{\ensuremath{\mathcal{L}_n}}
\newcommand{\aguni}{\ensuremath{\mathcal{AU}}}

\newcommand{\voca}{\ensuremath{\mathcal{T}}}

\newcommand{\intsys}{\ensuremath{\mathcal{I}}}

\newcommand{\snv}{\ensuremath{\mathit{SN}}}

\newcommand{\ag}{\ensuremath{\mathit{Ag}}}

\newcommand{\kbagent}{\ensuremath{\mathit{kb}}}
\newcommand{\Axioms}{\ensuremath{\Delta}}

\newcommand{\ekb}{\ensuremath{\text{EKB}}}
\newcommand{\ekbs}{\ensuremath{\text{EKBs}}}
\newcommand{\vekb}[2]{\ensuremath{\mathit{EKB}^{#1}_{#2}}}

\newcommand{\pre}{\ensuremath{\mathit{pred}}}
\newcommand{\nex}{\ensuremath{\mathit{next}}}

\newcommand{\kdfourfiveaxioma}{\textbf{KD45}}

\newcommand{\sfiveaxioma}{\textbf{S5}}

\newcommand{\Know}[2]{\ensuremath{K^{#1}_{#2}}}
\newcommand{\Learnt}[2]{\ensuremath{L^{#1}_{#2}}}
\newcommand{\Forget}[2]{\ensuremath{F^{#1}_{#2}}}

\newcommand{\Believe}[2]{\ensuremath{B^{#1}_{#2}}}
\newcommand{\LearnBelieve}[2]{\ensuremath{A^{#1}_{#2}}}
\newcommand{\ForgetBelieve}[2]{\ensuremath{R^{#1}_{#2}}}

\newcommand{\SKnows}[2]{\ensuremath{S^{#1}_{#2}}}
\newcommand{\EKnows}[2]{\ensuremath{E^{#1}_{#2}}}

\newcommand{\actionPredicate}[2]{\ensuremath{a^{#1}_{#2}}}
\newcommand{\Permitted}[2]{\ensuremath{P^{#1}_{#2}}}

\newcommand{\connectionPredicate}[2]{\ensuremath{c^{#1}_{#2}}}

\newcommand{\predicate}[1]{\ensuremath{\vpred^{#1}}}

\newcommand*{\tfppf}{\ensuremath{\mathcal{PPF_T}}}
\newcommand*{\tkbl}{\ensuremath{\mathcal{KBL_T}}\xspace}

\newcommand*{\tppl}{\ensuremath{\mathcal{PPL_{T}}}}

\newcommand*{\fancy}[1]{\ensuremath{\mathcal{#1}}}
\newcommand*{\formulae}{\fancy{F}}

\newcommand*{\pol}[2]{\ensuremath{\llbracket #1 \rrbracket_\text{#2}}}

\newcommand*{\then}{\ensuremath{\implies}}
\newcommand*{\al}{\ensuremath{\Box}}
\newcommand*{\ev}{\ensuremath{\Diamond}}

\newcommand*{\tuple}[1]{\ensuremath{\langle #1 \rangle}}

\newcommand{\wftraceset}{\ensuremath{\mathcal{WFT}}}

\newcommand{\rtppf}{\ensuremath{\ppf_{\mathcal{RT}}}}
\newcommand{\rtfppf}{\ensuremath{\rtppf}}

\newcommand{\rtsn}{\ensuremath{\mathcal{SN_{RT}}}}
\newcommand{\rtkbl}{\ensuremath{\kbl_{\mathcal{RT}}}}
\newcommand*{\frtkbl}{\ensuremath{\mathcal{F_{\rtkbl}}}}

\newcommand{\rtppl}{\ensuremath{\ppl_{\mathcal{RT}}}}

\newcommand*{\rtsnm}{SNM}

\newcommand*{\snm}{SNM}
\newcommand*{\snms}{SNMs}

\newcommand{\alice}{\ensuremath{\mathit{Alice}}}
\newcommand{\bob}{\ensuremath{\mathit{Bob}}}
\newcommand{\charlie}{\ensuremath{\mathit{Charlie}}}

\newcommand{\loc}{\ensuremath{\mathit{loc}}}

\newcommand{\post}{\ensuremath{\mathit{post}}}

\newcommand{\Friendship}{\ensuremath{\mathit{Friendship}}}
\newcommand{\friendship}{\ensuremath{\mathit{friendship}}}

\newcommand{\friendRequest}{\ensuremath{\mathit{friendRequest}}}
\newcommand{\FriendRequest}{\ensuremath{\mathit{FriendRequest}}}
\newcommand{\friend}{\ensuremath{\mathit{friend}}}

\newcommand{\event}{\ensuremath{\mathit{event}}}

\newcommand{\pic}{\ensuremath{\mathit{picture}}}

\newcommand{\acceptFollowReq}{\ensuremath{\mathit{acceptFollowReq}}}

\newcommand{\share}{\ensuremath{\mathit{share}}}

\newcommand{\disallowLoc}{\ensuremath{\mathit{disallowLoc}}}

\newcommand{\pub}{\ensuremath{\mathit{pub}}}
\newcommand{\weekend}{\ensuremath{\mathit{weekend}}}
\newcommand{\work}{\ensuremath{\mathit{work}}}
\newcommand{\dr}{\ensuremath{\mathit{DR}}}
\newcommand{\premise}{\ensuremath{\mathrm{PREMISE}}}

\newcommand{\enter}{\ensuremath{\mathsf{enter}}}
\newcommand{\eightpm}{\ensuremath{20\colon00}}
\newcommand{\sevenpm}{\ensuremath{19\colon00}}
\newcommand{\tenpm}{\ensuremath{22\colon00}}

\newcommand{\conservative}{\ensuremath{\mathit{conservative}}}
\newcommand{\susceptible}{\ensuremath{\mathit{susceptible}}}

\newcommand{\window}{\ensuremath{\omega}}
\newcommand{\agentype}{\ensuremath{\beta}}

\newcommand*{\valice}{\mword{Alice}}
\newcommand*{\vbob}{\mword{Bob}}

\newcommand*{\vdiane}{\mword{Diane}}

\newcommand*{\vgr}{\mword{G}}

\newcommand*{\vts}{\ensuremath{t}}

\newcommand*{\vaga}{\ensuremath{a}}
\newcommand*{\vagb}{\ensuremath{b}}

\newcommand*{\vagi}{\ensuremath{i}}
\newcommand*{\vagj}{\ensuremath{j}}

\newcommand*{\vwindow}{\ensuremath{w}}

\newcommand*{\vag}{\vagi}
\newcommand*{\vtrace}{\ensuremath{\sigma}}

\newcommand*{\vrtsn}{\mword{SN}}

\newcommand*{\vevs}{\ensuremath{E}}
\newcommand*{\vpolicy}{\ensuremath{\delta}}

\newcommand*{\vfa}{\ensuremath{\phi}}
\newcommand*{\vfb}{\ensuremath{\psi}}
\newcommand*{\vfc}{\ensuremath{\gamma}}
\newcommand*{\vf}{\vfa}
\newcommand*{\vevt}{\ensuremath{e}}
\newcommand*{\vevts}{\ensuremath{E}}

\newcommand*{\vnata}{\ensuremath{k}}

\newcommand*{\vnat}{\vnata}
\newcommand*{\vitera}{\ensuremath{i}}
\newcommand*{\viterb}{\ensuremath{j}}
\newcommand*{\viter}{\vitera}

\newcommand*{\vstart}{\ensuremath{s}}

\newcommand*{\vpred}{\ensuremath{p}}
\newcommand*{\term}{\ensuremath{s}}
\newcommand*{\vterm}{\vv{\term}}
\newcommand*{\mword}[1]{\ensuremath{\mathit{#1}}}
\newcommand*{\traceset}{\mword{TCS}}
\newcommand*{\tracets}{\ensuremath{\timestampset_\vtrace}}
\newcommand*{\trans}[1]{\ensuremath{\xrightarrow{#1}}}
\newcommand*{\mts}[1]{\ensuremath{\text{#1}}}
\newcommand*{\occur}[1]{\mathit{occurred}^{#1}}

\newcommand{\domSet}{\ensuremath{\mathcal{D}}}
\newcommand*{\frtppl}{\fancy{F_{PPL_{RT}}}}

\usepackage[sharp]{easylist}

\usepackage{appendix}

\usepackage{tabularx}

\let\temp\phi
\let\phi\varphi
\let\varphi\temp

\usepackage{pifont}

\begin{document}

\title{Timed Epistemic Knowledge Bases for Social Networks}
\subtitle{(Extended Version)}

\author{Ra\'ul Pardo}
\affiliation{%
  \institution{Chalmers $\mid$ University of Gothenburg}
  \city{Gothenburg}
  \state{Sweden}
}
\email{pardo@chalmers.se}

\author{C\'esar S\'anchez}
\affiliation{%
  \institution{IMDEA Software}
  \city{Madrid}
  \state{Spain}
}
\email{cesar.sanchez@imdea.org}

\author{Gerardo Schneider}
\affiliation{%
  \institution{Chalmers $\mid$ University of Gothenburg}
  \city{Gothenburg}
  \state{Sweden}
}
\email{gersch@chalmers.se}

\renewcommand{\shortauthors}{R. Pardo et al.}


\begin{abstract}
    We present an epistemic logic equipped with time-stamps in the atoms
  and epistemic operators, which allows to reason not only about
  information available to the different agents, but also about the
  moments at which events happen and new knowledge is acquired or
  deduced.
  Our logic includes both an epistemic operator and a belief operator,
  which allows to model the disclosure of information that may not be
  accurate.

  Our main motivation is to model rich privacy policies in online
  social networks.
  \emph{Online Social Networks} (OSNs) are increasingly used for
  social interactions in the modern digital era, which bring new
  challenges and concerns in terms of privacy.
  Most social networks today offer very limited mechanisms to express
  the desires of users in terms of how information that affects their
  privacy is shared.
  In particular, most current privacy policy formalisms allow only
  \emph{static policies}, which are not rich enough to express timed
  properties like ``\emph{my location after work should not be
  disclosed to my boss}''.
  The logic we present in this paper enables to express rich
  properties and policies in terms of the knowledge available to the
  different users and the precise time of actions and deductions.
  Our framework can be instantiated for different OSNs, by specifying
  the effect of the actions in the evolution of the social network and
  in the knowledge disclosed to each agent.

  We present an algorithm for deducing knowledge, which can also be
  instantiated with different variants of how the epistemic
  information is preserved through time.
  Our algorithm allows to model not only social networks with eternal
  information but also networks with ephemeral disclosures.
  Policies are modelled as formulae in the logic, which are interpreted
  over timed traces representing the evolution of the social network.
%

\end{abstract}

\maketitle

%
%

\section{Introduction}
\label{sec:intro}


{\em Online Social Networks} also known as {\em Social Networking
  Sites}, like Facebook~\cite{FacebookWeb}, Twitter~\cite{TwitterWeb}
and Snapchat~\cite{SnapchatWeb} have exploded in popularity in recent
years.
According to a recent survey~\cite{SNSuse} nearly~70\% of the Internet
users are active on social networks.

Some concerns, including privacy, have arisen alongside this
staggering increase in usage.
Several studies~\cite{MJB12spse+,JEM12fpc,YKBA11afps+,MJB11fosn+}
report that privacy breaches are growing in number.
Currently, the most popular social networks do not offer mechanisms
that users can use to guarantee their desired privacy effectively.
Moreover, virtually all privacy policies are static and cannot express
timing preferences, including referring to temporal points explicitly
and policies that evolve in time.

In~\cite{PKSS16sepposn} we presented a privacy policy framework able to express
dynamic privacy policies by introducing an explicit learning operator
and time intervals in the semantics of policies.
%
%
The framework consists of a {\em knowledge based logic} to characterise what users in
the social network know, and
a {\em privacy policy language}, based on the previous logic, where users can limit
who can know their information and when.
Policies and formulae in the logic are interpreted over {\em social network models} which faithfully represent the social graph of OSNs.
The policy language allows for the representation of recurrent privacy policies, .e.g.,
\emph{"During the weekend only my friends can see my pictures"}.
Note that the previous policy only requires to activate the static policy
\emph{"only my friends can see my pictures"} during weekends.
%

Though quite expressive, a major restriction of the logic in~\cite{PKSS16sepposn} is that it does not have {\em time}, so one cannot express different instants at which an
event happens and when knowledge is acquired.
%
Moreover, the logic only includes a knowledge modality thus implicitly assuming that
the information that users are told is true.
This assumption is, however, not realistic in social networks as
users may also have \emph{beliefs} given that not all sources of information are
truthful. Information that users disclose might be false.
There exists a growing interest in the detection of fake
news~\cite{fakenews1,fakenews2,fakenews3}.

To address this issues we propose to define a logic that:
\begin{inparaenum}[i)]
  \item is tailored for social networks, i.e., allows for expressing properties
  based on the social connections between users;
  \item combines knowledge and belief modalities to differentiate between
  beliefs that might be false and true knowledge;
  \item has time-stamps in modalities and atoms, thus enabling the possibility
  of referring to the time of information, e.g., a weekend location, and the
  moment when it was learnt.
\end{inparaenum}

There exist already some logics that include these elements---but
separately.
One line of work studies logics where the belief modality and atoms are
time-stampted~\cite{ZDDT15arbat}.
Unfortunately, this logic lacks of a knowledge modality, and is not tailored for
social networks.
Instead, it is defined to reason about AGM belief revision, which is a more
general setting than privacy in social networks.
A logic to reason about how beliefs spread out in Twitter has been introduced
in
\cite{XASZ17tlt}.
This logic, though tailored for social networks, does not include time-stamps,
so it cannot be used for reasoning about time.
Finally,
\cite{HSS09dktb} propose an axiomatisation of epistemic
logic which combines knowledge and belief, but again, it
does not contain time-stamps in modalities or atoms.

In this paper we leverage the insights from previous work to define a logic that
combines knowledge, belief and time and is tailored for the purpose of
specifying {\em dynamic privacy policies} for social networks.
We refer the reader to Section~\ref{sec:related} for a more detailed comparison
of related works.

More concretely, we extend~\cite{PKSS16sepposn} by enriching the logic with
explicit time instants, both in the atoms and in the epistemic operators.
In the resulting logic, one can refer to the instant at which some
knowledge is inferred, for example about the knowledge of another
agent at another instant.
The expressive power of the new logic allows to derive the
\emph{learning} operator from the time-stamped knowledge.

Second, we equip the logic with {\em belief} operators, with the only
restriction that agents cannot believe in something that they know is false, and
we define a belief propagation algorithm which guarantees that agents' beliefs
are always consistent.
This allows the instantiation of the framework to OSNs where gossiping
is allowed, that is, the spreading of potentially false information.
%
Analogous to the learning operator, we derive the \emph{accept}
operator as the moment in which an agent start believing in something.

Third, we introduce the notion of \emph{extended knowledge bases}
which allow to answer queries of (temporal) epistemic formulas against
the knowledge acquired during a sequence of events.
The algorithm uses epistemic deductive reasoning starting from a set
of axioms, which include the corresponding conventional axioms of
epistemic reasoning populated for all instants.
Depending on the desired instantiation of the framework, the axiom of
perfect recall or weaker versions of it can be instantiated which
allows to model knowledge acquisition in eternal OSNs like Facebook
and ephemeral OSNs like Snapchat.
When weaker versions are used, one can derive operators that capture
when an agent stops knowing something (the {\em forget} operator) or stops
believing in something (the {\em reject} operator).
Using extended knowledge bases we define the semantics of the logic
used to construct dynamic privacy policies.
This provides the means for a particular OSNs to enforce privacy
policies, for example by blocking an event that would result in a
violation of a user's policy.
We illustrate with examples how our rich logic allows to express
privacy policies and how agents can infer knowledge in different
instantiations of the framework.

Fourth, We prove that the model checking problem for this logic is decidable by
providing a model checking algorithm that is also used to check policy
conformance.
As a result, policy conformance is also decidable.

The purpose of this paper is to present an expressive foundation for developing
algorithms for detecting privacy violations and for privacy enforcement.
We leave the discussion of specialized efficient algorihtms for future work.

The rest of the paper is organised as follows.
Section~\ref{sec:real-time-rtppf} presents the framework and
Section~\ref{sec:syntax-of-rtkbl:rt} introduces the logic \rtkbl.
Section~\ref{sec:rtppl} shows how to express privacy policies using \rtkbl.
Section~\ref{sec:rtppl-examples} shows how different existing
OSNs can be modeled within our framework, including Snapchat.
Finally, Section~\ref{sec:related} presents related work and
Section~\ref{sec:conclusion} includes concluding remarks.
%

%
%

\section{A Timed Privacy Policy Framework}
\label{sec:real-time-rtppf}
%
We introduce a formal privacy policy framework for OSNs where
properties regarding users knowledge and beliefs as well as time can
be expressed.
Our framework extends the framework in~\cite{PKSS16sepposn}, in which
temporal properties were expressed using temporal operators \al\ and
\ev.
Our solution allows to describing properties at concrete moments in
time.
We also introduce a modality for beliefs which allows to distinguish
between information that might be inaccurate.

Our framework, called \rtppf, consists of the following components:
\begin{enumerate}
\item A timed epistemic logic, \rtkbl, where modalities and predicates
  are timestamped.
\item Extended social graphs, called \emph{social network models}, which
  describe the state of the OSN.
  These graphs contain the users or \emph{agents} in the system and
  the relations between them, and their knowledge and beliefs.
  We use $\rtsn$ to denote the universe of social network models, and
  use the notion of trace of social network models to describe the
  evolution of the system.
\item A parameter $\window \in \nat$ which determines for how long
  users remember information, e.g., $5$ seconds, $24$ hours or
  indefinitely.
  Users can acquire new believes that challenge their current
  knowledge and believes, which requires a resolution.
  We illustrate this by a parameter $\agentype$, with two
  possibilities \emph{conservative} or \emph{susceptible}, which
  specifies how users behave when they learn new beliefs which are
  inconsistent with their current beliefs.
\item A privacy policy language based on the previous logic.
\end{enumerate}



\section{A Timed Knowledge Based Logic}
\label{sec:syntax-of-rtkbl:rt}

\rtkbl\ is a knowledge based first order logic which borrows
modalities from epistemic logic~\cite{FHM+95rk}, equips modalities and
predicates with time-stamps, and allows quantifiers over time-stamps.

\subsection{Syntax}

Let \voca~be a \emph{vocabulary} which consists of a set of predicate
symbols, function symbols, and constant symbols.
Predicate and functions symbols have some implicit arity.
We assume an infinite supply of variables $x, y,\ldots$.
Terms of the elements of \voca~ can be built as $\term ::= c \mid x
\mid f(\vterm)$ where \vterm~is a tuple of terms respecting the arity
of $f$.

Let \timestampset\ denote a set of \emph{time-stamps}, which is required
to be a non-Zeno totally ordered set, that is, there is a finite
number of instants between any two given instants.
We use time-stamps to mark pieces of information or to query the
knowledge of the agents at specific points in time.
We consider \ag{} be a set of agents, \domSet\ a set of domains, and
use \evtSet{} for set of events that can be executed in a social
network.
For instance, in Facebook, users can share posts, upload pictures,
like comments, etc.
The set of events that users can perform depends on the social
network.
Similarly, we use $\conSet$ and $\actSet$ to denote special sets of
predicate symbols that denote connections (between agents) and
permissions.
We introduce the syntax of \rtkbl~as follows:

\begin{definition}[Syntax of \rtkbl{}]
\label{def:syntax-of-rtkbl:rt}
Given agents $ \vagi, \vagj \in \ag $ a time-stamp $ \vts \in
\timestampset$, an event $ \vevt \in \evtSet $, a variable $ x $, a
domain $D \in \domSet$, predicate symbols $ \connectionPredicate{\vts}{c}(\vagi, \vagj),
\actionPredicate{\vts}{a}(\vagi, \vagj), \predicate{\vts}(\vterm)$ where $c \in \conSet $ and $a
\in \actSet $, the \emph{syntax of the real-time knowledge-based logic
  \rtkbl{}} is inductively defined as:
   \begin{align*}
     \phi &::= \rho \mid
     \phi \land \phi \mid
     \neg \phi \mid
     \forall t \cdot \phi \mid
     \forall x : D \cdot \phi \mid
     \Know{\vts}{\vagi} \phi \mid
     \Believe{\vts}{\vagi} \phi
     \\
     \rho &::= \connectionPredicate{\vts}{c}(\vagi, \vagj) \mid
     \connectionPredicate{\vts}{a}(\vagi, \vagj) \mid
     \predicate{\vts}(\vterm) \mid
     \occur{\vts}(\vevt)
   \end{align*}
   Given a nonempty set of agents $ \vgr \subseteq \ag $, the
   additional epistemic modalities are defined $\SKnows{\vts}{G} \phi
   \triangleq \bigvee_{i \in G} \Know{\vts}{i} \phi $, $
   \EKnows{\vts}{G} \phi \triangleq \bigwedge_{i \in G} \Know{\vts}{i}
   \phi$.
 \end{definition}

The epistemic modalities stand for:
\begin{inparaenum}
\item[$\Know{\vts}{\vagi} \phi$,] agent $i$ knows $\phi$ at time $\vts$;
\item[$S^{\vts}_G \phi$,] someone in the group $G$ knows $\phi$ at time $\vts$;
\item[$E^{\vts}_G \phi$,] everyone in the group $G$ knows $\phi$ at time $\vts$.
\end{inparaenum}
We use the following notation as syntactic sugar $P^j_ia^t\triangleq
a(i,j,t)$ , meaning that ``\emph{agent $i$ is permitted to execute
  action $a$ to agent $j$}''.
For example, $P_\bob^\alice \mathit{friendRequest}^{5}$ means that Bob
is allowed to send a friend request to Alice at time $5$.
We will use \frtkbl{} to denote the set of all well-formed \rtkbl{}
formulae.
The syntax introduces the following novel notions that have not been
considered in other formal privacy policies languages such as
\cite{F11rbacppl,BFS+12rbaceehl,PardoLic,PKSS16sepposn}.

\textit{Time-stamped Predicates}.
Time-stamps are explicit in each predicate, including connections and
actions.
A time-stamp attached to a predicate captures those moment in time when
that particular predicate holds.
For instance, if Alice and Bob were friends in a certain time period,
then the predicate $\friend^{\vts}(\valice, \vbob)$ is true for all
\vts{} falling into the period, and false for all \vts{} outside.
This can be seen as the \emph{valid time} in temporal
databases~\cite{temporaldatabases}.

\textit{Separating Knowledge and Belief}.
Not all the information that users see in a social network is true.
For instance, Alice may tell Bob that she will be working until late,
whereas she will actually go with her colleagues to have some beers.
In this example, Bob has the (false) belief that Alice is working.

Traditionally, in epistemic logic, the knowledge of agents consists on
true facts.
Potentially false information is regarded as beliefs~\cite{FHM+95rk}.
The set of axioms \sfiveaxioma{} characterise knowledge and the axioms
of \kdfourfiveaxioma~characterise belief~\cite{FHM+95rk}.
For \rtkbl~we combine both notions in one logic.
In the following section we describe how to combine these two
axiomatisations based on the results proposed by Halpern \etal
in~\cite{HSS09dktb}.

\textit{Time-stamped Epistemic Modalities}.
Time-stamps are also part of the epistemic modalities $\Know{}{}$ and
$\Believe{}{}$.
Using time-stamps we can refer to the knowledge and beliefs of the
agents at different points in time.
For example, the meaning of the formula $\Believe{\eightpm}{\bob}
\loc^{\sevenpm}(\alice,\work)$ is that Bob beliefs at \eightpm{} that
Alice's location at \sevenpm{} is work.

\textit{Occurrence of Events}.
It is important to be able to determine when an event has occurred.
Such an expressive power allows users to define policies that are
activated whenever someone performs an undesired event.
Examples of these policies are: ``if Alice unfriends Bob, she is not
allowed to send Bob a friend request'' or ``if a Alice denies an
invitation to Bob's party, then she cannot see any of the pictures
uploaded during the party.''

Here we introduce $\occur{\vts}(\vevt)$ to be able to syntactically
capture the moment when a specific event $e$ occurred.
A similar predicate was introduced by Moses \etal
in~\cite{ben2013agent} for analysing communication protocols.
%

\subsection{Semantics}

\subsubsection{Real-Time Social Network Models}
\label{sec:rtsnm}

We introduce formal models which allow us to reason about specific social network states at a given moment in time.
These models leverage the information in the social graph~\cite{K14cn}---the core data model in most social networks~\cite{twitterSocialGraph,facebookSocialGraph,linkedinSocialGraph}.
Social graphs include the users (or \emph{agents}) and the relationships between them.
Moreover, in our models we include a knowledge base for each agent, and the set of privacy policies that they have activated.
We reuse the models defined for the previous version of this framework~\cite{PKSS16sepposn}.
Nevertheless, the expressiveness of the privacy policies that can be enforced in \rtppf~have substantially increased (see Section~\ref{sec:rtppl}).

\begin{definition}[Social Network Models]
\label{def:rtsnm}
   Given
   a set of formulae $\formulae \subseteq \frtkbl$,
   a set of privacy policies \poSet, 
   and a finite set of agents $ \ag \subseteq \aguni $ from a universe \aguni,
   a \emph{social network model} (\snm) is a tuple
   $ \langle \ag, \struct, \kb, \poval \rangle $,
   where
   \begin{inparaitem}

      \item
         \ag~is a nonempty finite set of nodes representing the agents in the
         social network;

      \item
         $\struct$ is a first-order structure over the \rtsnm. As usual, it
         consists of a set of domains, and a set relations, functions and
         constants interpreted over their corresponding domain.

      \item
         $ \kb: \ag \to 2^{\formulae} $ is a function retrieving a set of knowledge of an agent---each piece with an associated time-stamp. The set corresponds to the facts stored in the knowledge base of the agent; we write $\kb_i$.
         for $ \kb(i) $;

      \item
         $ \pi: \ag \to 2^\poSet $ is a function returning the set of privacy
         policies of each agent; we write $\poval_i $ for $ \poval(i)$.

   \end{inparaitem}
\end{definition}
In Def.~\ref{def:rtsnm}, the shape of the relational structure $\struct$ depends on the type of the social network under consideration.
We represent the connections---edges of the social graph---and the permission actions between social network agents, as families of binary relations, respectively $ \{C_i\}_{i \in \conSet} \subseteq \ag \times \ag $ and $ \{A_i\}_{i \in \actSet} \subseteq \ag \times \ag $ over the domain of agents.
We use $\{D_i\}_{i \in \domSet}$ to denote the set of domains.
The set of agents $\ag$ is always included in the set of domains.
We use $\conSet, \actSet \text{ and } \domSet$ to denote sets of connections, permissions and domains, respectively.

\subsubsection{Evolution of Social Network Models}
\label{sec:evolving-osns:rt}

The state of a social network changes by means of the execution of \emph{events}.
For instance, in Facebook, users can share posts, upload pictures, like comments, etc.
The set of events that users can perform depends on the social network.
We denote the set of events that can be executed in a social network as \evtSet.
We use traces to capture the evolution of the social network.
Each element of the trace is a tuple containing:
a social network model,
a set of events, and
a time-stamp.

\begin{definition}[Trace]
\label{def:trace:rt}
   Given $ \vnat \in \nat $, a trace $ \vtrace $ is a finite sequence
   $$ \vtrace = \tuple{(\vrtsn_0, \vevts_0, \vts_0), (\vrtsn_1, \vevts_1, \vts_1), \dots, (\vrtsn_\vnat, \vevts_\vnat, \vts_\vnat)} $$
   such that, for all $ 0 \leq \viter \leq \vnat $, $ \vrtsn_\viter \in \rtsn $, $
   \vevts_\viter \subseteq \evtSet $, and $ \vts_\viter \in \timestampset $.

\end{definition}

We define $ \tracets = \{ \vts \mid (\vrtsn, \vevts, \vts) \in \vtrace \} $ to be the set of all the time-stamps of $\vtrace$.
We impose some conditions to traces so that they accurately model the evolution of social networks.
We say that a trace is \emph{well-formed} if it satisfies the following conditions:

\paragraph{Ordered time-stamps.}
Time-stamps are strictly ordered from smallest to largest.

\paragraph{Accounting for Events.}
The definition has to account for events being explicit in the trace.
Let $\trans{}$ be a transition relation defined as $\trans{} \; \subseteq \rtsn \times 2^\evtSet \times \timestampset \times \rtsn $.
We have $\tuple{\vrtsn_1, \vevs, \vts, \vrtsn_2} \in \; \trans{} $ if $ \vrtsn_2 $ is the result of the set of events $ \vevs \in \evtSet $ happening in $ \vrtsn_1 $ at time \vts{}.
Note that we allow \vevs{} to be empty, in which case $ \vrtsn_2 = \vrtsn_1 $.
We will use the more compact notation of $\vrtsn_1 \trans{\vevs, \vts} \vrtsn_2$ where appropriate.

\paragraph{Events are independent.}
For each $\trans{E,\vts}$ the set of events $E$ must only contain independent events.
Two events are independent if, when executed sequentially, the execution order does not change their behaviour.
Consider the events: $\post(\charlie,\bob,"London")$ (Charlie shares a post containing Bob's location), and $\friendRequest(\alice,\charlie)$ (Alice sends a friend request to Charlie).
Independently of the order in which the previous events are executed the resulting \rtsnm\ will have a new post by Bob, and Charlie will receive a friend request from Alice.
On the other hand, consider now: $\post(\charlie, \bob,``London'')$ and $\disallowLoc(\bob)$ (Bob activates a privacy policy which forbids anyone to disclose his location).
In this case, if $\post(\charlie, \bob,``London'')$ is executed first, the resulting \rtsnm\ will contain the post by Charlie including Bob's location.
However, if $\disallowLoc(\bob)$ occurs before $\post(\charlie, \bob,``London'')$, Charlie's post would be blocked---since it violates Bob's privacy policy.
These two events are not independent.

More formally,

\begin{definition}
 Given two events $e_1, e_2 \in \evtSet$, we say that $e_1$ and $e_2$ are independent iff for any two traces $\vtrace_1$ and $\vtrace_2$
  $$\vtrace_1 = \snv^{0}_1 \trans{\{e_1\}, \vts} \snv^{1}_1 \trans{\{e_2\}, \vts'}  \snv^{2}_1$$
  $$\vtrace_2 = \snv^{0}_2 \trans{\{e_2\}, \vts} \snv^{1}_2 \trans{\{e_1\}, \vts'}  \snv^{2}_2$$
  it holds $\snv^{0}_1 = \snv^{0}_2$ and $\snv^{2}_1 = \snv^{2}_2$.
\end{definition}

We can now provide a formal definition of well-formed \snm\ traces.

\begin{definition}[Well-Formed Trace]
\label{def:well-formed-trace:rt}
   Let \[ \vtrace = \tuple{(\vrtsn_0, \vevts_0, \vts_0),
                           (\vrtsn_1, \vevts_1, \vts_1), \dots,
                           (\vrtsn_\vnat, \vevts_\vnat, \vts_\vnat)} \]
   be a trace. \vtrace{} is \emph{well-formed} if the following conditions hold:
   \begin{enumerate}
      \item
         For any $ \vitera, \viterb $ such that $ 0 \leq \vitera < \viterb \leq \vnat $ it follows that $\vts_\vitera < \vts_\viterb $.
      \item
         For all \viter{} such that $ 0 \leq \viter \leq \vnat - 1 $, it is the case that $ \vrtsn_\viter \trans{\vevs_{i + 1}, \; \vts_{\viter + 1}} \vrtsn_{\viter + 1} $.
      \item
        For all $e_1, e_2 \in E_i$ for $ 0 \leq i \leq k$, $e_1$ is independent from $e_2$.
   \end{enumerate}
\end{definition}

We will use \traceset{} to refer to the set of all well-formed \rtfppf{} traces.
In order to be able to syntactically refer to the previous or next social network model, given a concrete time-stamp, we assume that there exist the functions \emph{predecessor} ($\pre$) and \emph{next} ($\nex$).
$\pre: \timestampset \rightarrow \timestampset$ takes a time-stamp and returns the previous time-stamp in the trace.
Since the set of time-stamps is non-Zeno it is always possible to compute the previous time-stamp.
Analogously, $\nex: \timestampset \rightarrow \timestampset$ takes a time-stamp and returns the next time-stamp in the trace.
In $\langle (\vrtsn_0, E_0, t_0), (\vrtsn_1, E_1, t_1), (\vrtsn_2, E_2, t_2) \rangle$, $\pre(t_1) = t_0$ and $\nex(t_1) = t_2$.
We define predecessor of the initial time-stamp to be equal to itsel, i.e., $\pre(t_0) = t_0$.
Similarly, next of the last time-stamp of the trace is equal to itself, i.e., $\nex(t_2) = t_2$.

\subsubsection{Modelling knowledge}

It is not a coincidence that \rtkbl\ formulae look very similar to those of the language $\langn$ originally defined for epistemic logic~\cite{FHM+95rk}.
We would like to provide users in our system with the same notion of knowledge.
Traditionally, in epistemic logic, the way to model and give semantics to $K_i \phi$ is by means of an undistinguishability relation which connects all world that an agent considers possible~\cite{FHM+95rk}.
In particular, when talking about traces of events the framework used is \emph{Interpreted Systems} (ISs).
In ISs traces are called \emph{runs}.
A run describes the state of the system at any point in (discrete) time.
An IS \intsys~is composed by a set of runs and a undistinguishability relation ($\sim_i$)---for each agent $i$---which models the states of the system that agents consider possible at any point in time.
Determining whether an agent $i$ knows a formula $\phi$ (written in the language of epistemic logic), for a run $r$ at time $m$ is defined as follows:
$$ (\intsys, r, m) \sat K_i \phi \text{ iff } (\intsys, r', m') \sat \phi \text{ for all } (r,m) \sim_i (r',m')$$
where $(r,m)$ and $(r',m')$ represent states of \intsys.

Additionally, Fagin~\etal proposed an alternative encoding  to answer epistemic queries from a knowledge base consisting in a set of accumulated facts~\cite{FHM+95rk}[Section 7.3].
Let $\kbagent$ be an agent representing a knowledge base that has been told the facts $\langle \psi_1, \ldots, \psi_k \rangle$ for $k \geq 1$ in run $r$ at time $t$.
It was shown in \cite{FHM+95rk}[Theorem 7.3.1] that the following are equivalent:
\begin{enumerate}[a)]
  \item $(\intsys^{\kbagent}, r, m) \sat K_{\kbagent} \phi$.
  \item $\mrstn \sat K_{\kbagent} (\psi_1 \wedge \ldots \wedge \psi_k) \implies K_{\kbagent} \phi$.
  \item $(\psi_1 \wedge \ldots \wedge \psi_k) \implies \phi$ is a tautology.
\end{enumerate}

where $\intsys^{\kbagent}$ is an IS which models the behaviour of $\kbagent$ and $\mrstn \sat \phi$ means that $\phi$ is valid in the Kripke models with an accessibility relation that is reflexive ($\mathit{r}$), symmetric ($\mathit{s}$) and transitive ($\mathit{t}$).
The previous theorem holds not only for a system consisting in a single knowledge base, but systems including several knowledge bases.

This way of modelling knowledge is very suitable for our social network models.
As mentioned earlier, in a social network model the users' knowledge is stored in their knowledge base.
Therefore, by using the equivalence in \cite{FHM+95rk}[Theorem 7.3.1] we can determine whether a user knows a formula $\phi$ from the conjunction of all the formulae it has been told, formally,
$ \bigwedge_{\psi \in \kb_i} \psi \implies \phi $.

However, as mentioned earlier, \rtkbl\ is not the same language as \langn.
Therefore we cannot directly apply \cite{FHM+95rk}[Theorem 7.3.1] to determine whether a user knows a fact $\phi$.
In the following we described an extended knowledge base which supports all the components of \rtkbl. 

\subsubsection{Extended Knowledge Bases}

An \emph{Extended Knowledge Base} (\ekb) consists in a collection \rtkbl~formulae without quantifiers.
All domains in a \rtsnm\ are finite at a given point in time---the might grow as events occur.
On the one hand, regular domains, i.e., $D^{\vts}$ in $\struct$ are assumed to be finite.
Therefore, they can be easily unfolded as a finite conjunction.
On the other hand, the time-stamps domain---though infinite in general---for a given trace $\timestampset_\vtrace$ will be finite because traces are finite.
Hence \ekbs\ can be populated with the explicit time-stamps values that moment in time.
Later in this section, we introduce some axioms which will define how time-stamps are handled.
In what follows we introduce the axioms \ekbs\ use to handle knowledge and belief.

\paragraph{Derivations in \ekbs}

The information stored in an agent's \ekbs\ along a trace determines her knowledge.
At a concrete moment in time, an agent's \ekb\ contains the explicit knowledge she just learnt.
New knowledge can be derived from the explicit pieces of information in agents' \ekbs.
Derivations are not limited to formulae of a given point in time, but also can use old knowledge.
A \emph{time window}, or simply, window, determines how much old knowledge is included in a derivation.
We write $\Gamma \der (\phi,\vwindow)$ to denote that $\phi$ can be derived from $\Gamma$ given a window $\vwindow$.
We provide a set of deduction rules, $\dr$, of the form
$$
\inferrule*
  {\Gamma \der (\phi, \vwindow')} 
  {\Gamma \der (\psi, \vwindow)} 
$$
meaning that, given the set of premises $\Gamma$, $\psi$ can be derived with a window $\vwindow$ from $\phi$ in a window $\vwindow'$.

\begin{definition}
  A {\em timed derivation} of a formula $\phi \in \wkbl$ given a window $\vwindow \in \nat$, is a finite sequence of pairs of formulae and windows, $\wkbl \times \nat$, such that $(\phi_1, \vwindow_1), (\phi_2, \vwindow_2), \ldots,$ $(\phi_n, \vwindow_n) = (\phi, \vwindow)$ where each $\phi_i$, for $1 \leq i \leq n$, follows from previous steps by an application of a deduction rule of $\dr$ which premises have already been derived, i.e., it appears as $\phi_j$ with $j < i$,
  and $\vwindow_j \leq \vwindow_i$.

  \label{def:derivation}
\end{definition}

In what follows we present the concrete derivation rules that can be used in \ekbs\ to derive knowledge.
We define deduction rules based on well studied axiomatisations of knowledge and belief together with rules to deal with knowledge propagation.

\paragraph{Knowledge and Belief in \ekbs}
\label{knowledge-belief-axioms}

In \ekbs\ knowledge and belief coexist.
So far, the definition of knowledge that we provided in the previous section only takes into account the axioms for knowledge.
In particular, the axiomatisation \sfiveaxioma.
Thus, \ekbs\ can use any of the \sfiveaxioma\ axioms to derive new knowledge from the conjunction of explicit facts in the knowledge base.

Fagin~\etal provided an axiomatisation for belief~\cite{FHM+95rk}, the \kdfourfiveaxioma\ axiomatisation.
It includes the same set of axioms as \sfiveaxioma~---replacing $\Know{}{i}$ by $\Believe{}{i}$---except for the axiom $\Know{}{i} \phi \implies \phi$ (A3).
The difference between knowledge and belief is that believes do not need to be true---as required by A3 in knowledge.
The requirement is that an agent must have \emph{consistent} beliefs.
It is encoded in the following axiom $\neg \Believe{}{i} \bot$ (axiom D).

We can summarise the last two paragraphs as follows: Whenever a formula of the form $\Know{}{i}\phi$ is encountered, axioms from \sfiveaxioma\ can be applied to derive new knowledge, and if the formula is of the form $\Believe{}{i}\phi$ axioms from \kdfourfiveaxioma\ can be used instead.
Nonetheless, we are missing an important issue: How do knowledge and belief relate to each other?
To answer this question we use two axioms proposed by Halpern~\etal in~\cite{HSS09dktb}:
(L1) $\Know{}{\vagi} \phi \implies \Believe{}{\vagi} \phi$ and
(L2) $\Believe{}{\vagi} \phi \implies \Know{}{\vagi} \Believe{}{\vagi} \phi$.


L1 expresses that when users know a fact they also believe it.
It is sound with respect to the definition of both modalities, since knowledge is required to be true by definition (recall axiom A3).
This axiom provides a way to convert knowledge to belief. 
L2 encodes that when agents believe a fact $\phi$ they know that they believe $\phi$.
Thus adding an axiom which introduces knowledge from belief---more precisely, introduces knowledge about the beliefs.

However, the axioms from \sfiveaxioma, \kdfourfiveaxioma\ and L1, L2 need to be adapted to \rtkbl\ syntax---which is the type of formulae supported by \ekbs.
In particular, the modalities need a time-stamp.
All these axiomatisations are defined for models which represent the system in a concrete time.
That is, given the current set of facts that users have, they can apply the axioms to derive new knowledge (at that time).
To preserve this notion we will simply add the time-stamp $t$ to all modalities.
Intuitively, it models that if users have some knowledge at time $t$ they can derive knowledge using the previous axioms, and, this derived knowledge, belongs to the same time $t$.
Table \ref{tab:ekb-axioms} shows the complete list of axioms that can be applied given a trace $\vtrace$ for each time-stamps $t \in \timestampset_\vtrace$.

\begin{table}[t]
  \begin{center}
    \begin{tabular}{|c|l|c|l|c|l|}
      \hline
      \multicolumn{2}{|c|}{\textbf{Knowledge axioms}} &
      \multicolumn{2}{c|}{\textbf{Belief axioms}} &
      \multicolumn{2}{c|}{\textbf{Knowledge-Belief axioms}} \\
      \hline
      A1 & All tautologies of first-order logic &
      K & $\Believe{\vts}{i} \phi \wedge \Believe{\vts}{i}(\phi \implies \psi) \implies \Believe{\vts}{i} \psi$ &
      L1 & $\Know{\vts}{\vagi} \phi \implies \Believe{\vts}{\vagi} \phi$ \\
      A2 & $\Know{\vts}{i} \phi \wedge \Know{\vts}{i}(\phi \implies \psi) \implies \Know{\vts}{i} \psi$ &
      D & $\neg \Believe{\vts}{i} \bot$ &
      L2 & $\Believe{\vts}{\vagi} \phi \implies \Know{\vts}{\vagi} \Believe{\vts}{\vagi} \phi$ \\
      A3 & $\Know{\vts}{i} \phi \implies \phi$ &
      B4 & $\Believe{\vts}{i} \phi \implies \Believe{\vts}{i} \Believe{\vts}{i} \phi$ &
      & \\
      A4 & $\Know{\vts}{i} \phi \implies \Know{\vts}{i} \Know{\vts}{i} \phi$ &
      B5 & $\neg \Believe{\vts}{i} \phi \implies \Believe{\vts}{i} \neg \Believe{\vts}{i} \phi$ & & \\
      A5 & $\neg \Know{\vts}{i} \phi \implies \Know{\vts}{i} \neg \Know{\vts}{i} \phi$ & & & & \\
      \hline
      \end{tabular}
      \caption{\ekb\ axioms for a trace $\vtrace$ for each $t \in \timestampset_\vtrace$.}

      \label{tab:ekb-axioms}
  \end{center}
\end{table}

In order for these axioms to be used in timed derivations we now express them as deduction rules as shown in Table~\ref{tab:kb-deduction}.
Since all derivations are performed for the same $\vts$ they all share the same window $\vwindow$.

\begin{table*}[t]
  \begin{center}
    \scalebox{0.85}{
    \begin{tabular}{|p{3.8cm} p{3.8cm} p{3.8cm} p{3.8cm}|}
      \hline
      \multicolumn{4}{|c|}{\textbf{Knowledge deduction rules axioms}} \\
      \hline
      &&&\\
      $\inferrule*[Right=(A1)]{\phi \text{ is a first-order tautology}}{\Gamma \der (\phi, \vwindow)}$ &
      \multicolumn{2}{c}{
      $\inferrule*[Right=(A2)]{\Gamma \der (\Know{\vts}{i} \phi, \vwindow) \\ \Gamma \der (\Know{\vts}{i}(\phi \implies \psi), \vwindow)}{\Gamma \der (\Know{\vts}{i} \psi, \vwindow)}$
      } &
      $\inferrule*[Right=(A3)]{\Gamma \der (\Know{\vts}{i} \phi, \vwindow)}{\Gamma \der (\phi, \vwindow)}$ \\
      &&&\\
      \multicolumn{2}{|c}{
        $\inferrule*[Right=(A4)]{\Gamma \der (\Know{\vts}{i} \phi, \vwindow)}{\Gamma \der (\Know{\vts}{i} \Know{\vts}{i} \phi, \vwindow)}$
      } &
      \multicolumn{2}{c|}{
        $\inferrule*[Right=(A5)]{\Gamma \der (\neg \Know{\vts}{i} \phi, \vwindow)}{\Gamma \der (\Know{\vts}{i} \neg \Know{\vts}{i} \phi, \vwindow)}$
      } \\
      &&&\\
      \hline
      \multicolumn{4}{|c|}{\textbf{Belief deduction rules}} \\
      \hline
      &&&\\
      $\inferrule*[Right=(K)]{\Gamma \der (\Believe{\vts}{i} \phi, \vwindow) \\ \Gamma \der (\Believe{\vts}{i}(\phi \implies \psi), \vwindow)}{\Gamma \der (\Believe{\vts}{i} \psi, \vwindow)}$ &
      $\inferrule*[Right=(D)]{~}{\Gamma \der (\neg \Believe{\vts}{i} \bot, \vwindow)}$ &
      $\inferrule*[Right=(B4)]{\Gamma \der (\Believe{\vts}{i} \phi, \vwindow)}{\Gamma \der (\Believe{\vts}{i} \Believe{\vts}{i} \phi, \vwindow)}$ &
      $\inferrule*[Right=(B5)]{\Gamma \der (\neg \Believe{\vts}{i} \phi, \vwindow)}{\Gamma \der (\Believe{\vts}{i} \neg \Believe{\vts}{i} \phi, \vwindow)}$ \\
      &&&\\
      \hline
      \multicolumn{1}{|c|}{\textbf{Premise deduction rule}} & \multicolumn{3}{c|}{\textbf{Knowledge-Belief deduction rules}} \\
      \hline
      \multicolumn{1}{|c|}{$~$}&&&\\
      \multicolumn{1}{|c|}{
      $\inferrule*[Right=(Premise)]{\phi \in \Gamma}{\Gamma \der (\phi, \vwindow)}$
      } &
      \multicolumn{2}{c}{
      $\inferrule*[Right=(L1)]{\Gamma \der (\Know{\vts}{i} \phi, \vwindow)}{\Gamma \der (\Believe{\vts}{i} \phi, \vwindow)}$
      } &
      $\inferrule*[Right=(L2)]{\Gamma \der (\Believe{\vts}{i} \phi, \vwindow)}{\Gamma \der (\Know{\vts}{i} \Believe{\vts}{i} \phi, \vwindow)}$\\
      \multicolumn{1}{|c|}{$~$}&&&\\
      \hline
    \end{tabular}}
      \caption{\ekb\ deduction rules for a trace $\vtrace$ for each $t \in \timestampset_\vtrace$.}
      \label{tab:kb-deduction}
  \end{center}

\end{table*}

As mentioned earlier, an \ekb\ models the knowledge and beliefs of a user at a given moment in time.
We encode this notion by adding an explicit $\Know{\vts}{i}$ to every formula in a user's \ekb.
Formally, we say that users in a trace \vtrace\ are \emph{self-aware} iff for all $\vts \in \timestampset_\vtrace$
$\text{If } \phi \in \vekb{\vtrace[\vts]}{i} \text{ then } \phi = \Know{\vts}{i} \phi'.$
By assuming this property we can syntactically determine the time $\vts$ when some knowledge $\phi$ enters an $\ekb$.
In what follows we assume that all well-formed traces are composed by self-aware agents.

\newcommand*\circled[2]{
  \tikz[baseline=(char.base)]{\node[shape=rectangle, draw, inner sep=2pt, rounded corners, label=below:#2] (char) {#1};}
}

\begin{example}
  \label{ex:deduction-ekb}
  Consider the following \ekb\ from a trace $\vtrace$ of an agent $\vagi$ at time $\vts$.
  \begin{center}
    \circled{
    \begin{tabular}{c}
      $\Know{\vts}{\vagi} (\forall \vts' \cdot \forall \vagj \colon \ag^{\vts'} \cdot \event^{\vts'}(\vagj,\pub) \implies \loc^{\vts'}(\vagj,\pub))$ \\
      $\Know{\vts}{\vagi} \event^{\vts}(\alice,\pub)$
    \end{tabular}}
    {$\vekb{\vtrace[\vts]}{\vagi}$}
  \end{center}

  In this \ekb\ $\vagi$ can derive using the axioms in Table~\ref{tab:ekb-axioms} that Alice's location at time $\vts$ is a pub, i.e., $\loc^{\vts'}(\alice,\pub)$.
  Here we show the steps to derive this piece of information.
  We recall that quantifiers are unfolded when added to the knowledge base.
  Therefore, given $\timestampset_\vtrace = \{\vts_0, \vts_1, \ldots, \vts\}$

  \begin{tabular}{l}
    $\Know{\vts}{\vagi} \forall \vagj \colon \ag^{\vts_0} \cdot \event^{\vts_0}(\vagj,\pub) \implies \loc^{\vts_0}(\vagj,\pub) ~ \wedge $ \\
    $\Know{\vts}{\vagi} \forall \vagj \colon \ag^{\vts_1} \cdot \event^{\vts_1}(\vagj,\pub) \implies \loc^{\vts_1}(\vagj,\pub) ~ \wedge $ \\
    $\ldots$ \\
    $\Know{\vts}{\vagi} \forall \vagj \colon \ag^{\vts} \cdot \event^{\vts}(\vagj,\pub) \implies \loc^{\vts}(\vagj,\pub)$
  \end{tabular}

  where each of these are also syntactic sugar, for instance, given $\ag^{\vts} = \{\vagj_1, \vagj_2, \ldots, \vagj_n\}$ for $n \in \nat$

  \begin{tabular}{l}
    $\Know{\vts}{\vagi} \event^{\vts}(\vagj_1,\pub) \implies \loc^{\vts}(\vagj_1,\pub) ~ \wedge $ \\
    $\Know{\vts}{\vagi} \event^{\vts}(\vagj_2,\pub) \implies \loc^{\vts}(\vagj_2,\pub) ~ \wedge $ \\
    $\ldots$\\
    $\Know{\vts}{\vagi} \event^{\vts}(\vagj_n,\pub) \implies \loc^{\vts}(\vagj_n,\pub) $
  \end{tabular}

  The predicate $\event^{\vts}(\vagj,\pub)$ means that $\vagj$ attended an event at time $\vts$ in a $\pub$.
  The predicate $\loc^{\vts}(\vagj)$ means that $\vagj$'s location is a $\pub$.
  Thus, the implication above encodes that if $\vagi$'s knows at $\vts$ that if an agent is attending an event in a pub at time $\vts$, her location will be a pub.
  Moreover, $\vagi$ knows at time $\vts$ that Alice is attending an event at the pub, $\event^{\vts}(\alice,\pub)$.
  As mentioned earlier, in epistemic logic, knowledge is required to be true.
  Therefore, $\event^{\vts}(\alice,\pub)$ must be a true predicate.
  From this we can infer that $\alice \in \ag^{\vts}$.
  Because of this, $\Know{\vts}{\vagi} \event^{\vts}(\alice,\pub) \implies \loc^{\vts}(\alice,\pub)$ must also be present in $\vekb{\vtrace[\vts]}{\vagi}$.
  Applying A2 to $\Know{\vts}{\vagi} \event^{\vts}(\alice,\pub)$ and the previous implication we can derive $\Know{\vts}{\vagi} \loc^{\vts'}(\vagj,\pub)$ as required.\qed
\end{example}

\paragraph{Handling time-stamps}
\label{handling-timestamps}

In \ekbs\ users can also reason about time.
For instance, if Alice learns Bob's birthday she will remember this piece of information, possibly, forever.
Nonetheless, this is not always true, some information is transient, i.e., it can change over time.
Imagine that Alice shares a post including her location with Bob.
Right after posting, Bob will know Alice's location---assuming she said the truth.
However, after a few hours, Bob will not know for sure whether Alice remains in the same location.
The most he can tell is that Alice was a few hours before in that location or that he believes that Alice's location is the one she shared.
We denote the period of time in which some piece of information remains true as \emph{duration}.

Different pieces of information might have different durations.
For example, someone's birthday never changes, but locations constantly change.
Duration also depends on the OSN.
In Snapchat messages last 10 seconds, in Whatsapp status messages last 24 hours and in Facebook posts remain forever unless a user removes them.
Due to these dependences we do not fix a concrete set of properties regarding time-stamps. 
Instead we keep it open so that they can be added when modelling concrete OSNs in \rtppf.
Concretely, the parameter $\window$ introduced at the beginning of Section~\ref{sec:real-time-rtppf} corresponds to information duration for a particular OSN modelled in \rtppf.

%

Using the window $\vwindow$---from timed derivations, see Def.~\ref{def:derivation}---we define the following deduction rule encoding knowledge propagation.
Given $\vts, \vts' \in \timestampset_\vtrace$ where $\vts < \vts'$:
$$\inferrule*[Right=(KR1)]{\Gamma \der (\Know{\vts}{i} \phi, \vwindow - (\vts' - \vts))}{\Gamma \der (\Know{\vts'}{i} \phi, \vwindow)}$$

The intuition behind $\text{KR1}$ is that $\vwindow$ is consumed every time knowledge is propagated.
Imagine that Alice knows at time 1 the formula $\phi$, $\Know{1}{\alice} \phi$.
Using $\text{KR1}$ in a derivation would allow us to derive, for instance, that she knows $\phi$ at a later time, e.g., at time 5, that is, $\Know{5}{\alice} \phi$.
Note that this derivation requires $\vwindow$ to be at least 4, since Alice's knowledge of $\phi$ is propagated 4 units of time.
As usual in when using this type rules, derivations can be described forwards or backwards.
In the latter, the derivation starts with the conclusion of the rule---that is, we want to derive---and reduce the value of $\vwindow$ every time we access old knowledge.
In the former, we start from the premise of the rule and we increase $\vwindow$ accordingly to derive the conclusion.
The intuition and ways of deriving knowledge using $\text{KR1}$ are better illustrated with an example.

\begin{example}
  Consider the sequence of \ekbs\ in Fig.~\ref{fig:seq-ekbs-2-knowledge-progation} of an agent $\vagi$ from a trace $\vtrace$ where $\timestampset_\vtrace = \{0, \ldots, 4\}$.
  In this example we show the purpose of the window $\vwindow$ when making derivations.
  Note that it is not possible to derive Alice's location only from the set of facts in a single knowledge base at a time $\vts$.
  Instead it is required to combine knowledge from different knowledge bases.
  We use the knowledge recall rule with different windows to access previous knowledge.
  As mentioned earlier, intuitively, $\vwindow$ determines for how long agents remember information.
  Therefore, it is required to find an appropriate value for $\vwindow$ that includes the sufficient knowledge---from all moments in time---to perform the derivation.
  In the figure, the red square marks the accessible knowledge for $\vwindow=2$ and the blue square for $\vwindow=3$.

  As can be seen in the trace, in order for $\vagi$ to derive $\event^3(\alice,\pub)$ she needs to combine knowledge from \vekb{\vtrace[0]}{\vagi} and \vekb{\vtrace[3]}{\vagi}.
  Let $\vekb{\vtrace}{\vagi} = \bigcup_{\vts \in \timestampset_\vtrace} \vekb{\vtrace[\vts]}{\vagi}$.
  First, we show how to construct a proof forwards, i.e., starting from the premises and a window of 0, move forward---by increasing the size of $\vwindow$---until the inference can be performed.
  In particular, we show that $\vekb{\vtrace}{\vagi} \der (\Know{3}{\loc^{3}(\alice,\pub)}, \vwindow)$ for $\vwindow \in \nat$.
  Let us start by applying the rule $\premise$ with $\vwindow = 0$,
  $$\vekb{\vtrace}{i} \der (\Know{0}{i} \event^{3}(\alice,\pub) \implies \loc^{3}(\alice,\pub),0)$$
  Recall that the quantifiers in the example are just syntactic sugar.
  They are replaced when added to the \ekb.
  Now we use $\text{KR1}$ to combine this knowledge with knowledge at time $3$.
  In other words, we propagate knowledge from time 0 (\Know{0}{\vagi}) to time 3 (\Know{3}{\vagi}).
  \begin{center}
    \begin{tabular}{c}
    $
    \inferrule*[Left=(KR1)]
    {\vekb{\vtrace}{i} \der (\Know{0}{i} \event^{3}(\alice,\pub) \implies \loc^{3}(\alice,\pub),0)}
    {\vekb{\vtrace}{i} \der (\Know{3}{i} \event^{3}(\alice,\pub) \implies \loc^{3}(\alice,\pub),3)}
    $
    \end{tabular}
  \end{center}
  As stated in the rule the window has been incresed by 3, since $0 = 3 - (3 - 0)$.
  We apply \premise\ again to obtain $(\vekb{\vtrace}{i} \der \Know{3}{i} \event^{3}(\alice,\pub),3)$.
  As in Example~\ref{ex:deduction-ekb} by applying A2 in the previous statements we derive $(\vekb{\vtrace}{i} \der \Know{3}{i} \loc^{3}(\alice,\pub),3)$.
  This proof shows that $i$ knows Alices location provided that agents remember information during 3 units of time.

  We show now that a window smaller than 3 makes this derivation impossible.
  Also we construct the proof backwards, i.e., starting from the conclusion we try to prove the required premises.
  Let us take a largest window smaller than 3, $\vwindow  = 2$.
  We try to show that $\vekb{\vtrace}{\vagi} \der (\Know{3}{\vagi} \loc^{3}(\alice,\pub),2)$.
  In order to prove we need to show the following:
  \begin{center}
  \begin{tabular}{c}
  $
  \inferrule*[Left=(A2)]
  {(\vekb{\vtrace}{i} \der \Know{3}{i} \event^{3}(\alice,\pub),2) \\ \vekb{\vtrace}{i} \der (\Know{3}{i} \event^{3}(\alice,\pub) \implies \loc^{3}(\alice,\pub),2)}
  {\vekb{\vtrace}{i} \der (\Know{3}{i} \loc^{3}(\alice,\pub),2)}
  $
  \end{tabular}
  \end{center}
  The first premise, $(\vekb{\vtrace}{i} \der \Know{3}{i} \event^{3}(\alice,\pub),2)$, trivally follows by \premise.
  For the second premise we first try move one step back using KR1, i.e,
  $$\vekb{\vtrace}{i} \der (\Know{2}{i} \event^{3}(\alice,\pub) \implies \loc^{3}(\alice,\pub),1),$$
  since there is no knowledge at time $2$, the previous statement cannot be proven.
  We apply again KR1 obtaining
  $$\vekb{\vtrace}{i} \der (\Know{1}{i} \event^{3}(\alice,\pub) \implies \loc^{3}(\alice,\pub),0).$$
  Similarly, this statement cannot be proven.
  Note that the window is $0$.
  Intuitively, it means that we have access all knowledge that $\vagi$ remembers.
  Therefore, we have reach a dead end in our proof tree.
  KR1 could be applied again, which would give
  $$\vekb{\vtrace}{i} \der (\Know{1}{i} \event^{3}(\alice,\pub) \implies \loc^{3}(\alice,\pub),-1).$$
  However, it would not be a valid timed derivation.
  In Def.~\ref{def:derivation} we require $\vwindow \in \nat$.
  Finally, we conclude that $\Know{3}{i} \loc^{3}(\alice,\pub)$ cannot be derived with a window of $2$.\qed
\end{example}

\begin{figure*}[!t]
  \centering
  \scalebox{0.8}{
  \footnotesize
  \begin{tikzpicture}[>=stealth',every text node part/.style={align=center},on grid,auto, minimum height=1ex,semithick, label distance=0.2cm]
    \node (ekb0) [rectangle, rounded corners, draw, minimum width=30pt, label=below:$\hspace{0.2cm}\vekb{\vtrace[0]}{\vagi}$]
                 {$\Know{0}{\vagi} \forall \vts' \cdot \forall \vagj \colon \ag^{\vts'} \cdot \event^{\vts'}(\vagj,\pub) \implies \loc^{\vts'}(\vagj,\pub)$};
    \node (ekb1) [rectangle, rounded corners, draw, minimum width=30pt, label=below:$\vekb{\vtrace[1]}{\vagi}$][right = 4.5cm of ekb0]
                 {$\emptyset$};
    \node (ekb2) [rectangle, rounded corners, draw, minimum width=30pt, label=below:$\vekb{\vtrace[2]}{\vagi}$][right = 1.4cm of ekb1]
                 {$\emptyset$};
    \node (ekb3) [rectangle, rounded corners, draw, minimum width=30pt, label=below:$\vekb{\vtrace[3]}{\vagi}$][right = 2.4cm of ekb2]
                 {$\Know{3}{\vagi} \event^{3}(\alice,\pub)$};
    \node (ekb4) [rectangle, rounded corners, draw, minimum width=30pt, label=below:$\vekb{\vtrace[4]}{\vagi}$][right = 2.5cm of ekb3]
                 {$\emptyset$};

    \node (w0) [draw=red, thick, fit= (ekb1) (ekb3), label=right:$~$] {};
    \node (w1) [draw=blue, thick, fit= (ekb0) (w0) , label=right:$~$] {};

    \path[->] [] (ekb0) edge node {} (ekb1);
    \path[->] [] (ekb1) edge node {} (ekb2);
    \path[->] [] (ekb2) edge node {} (ekb3);
    \path[->] [] (ekb3) edge node {} (ekb4);
  \end{tikzpicture}
  } 

  \caption{Sequence of \ekbs\ of agent $\vagi$ for a trace $\vtrace$ where $\timestampset_\vtrace = \{0, \ldots, 4\}$}
  \label{fig:seq-ekbs-2-knowledge-progation}
\end{figure*}

\paragraph{Belief propagation}
Beliefs cannot be propagated as easily as knowledge because new beliefs may
contradict current knowledge or beliefs of an agent.
Instead of using timed derivations, we model agents that try to propagate
beliefs if these beliefs are consistent, and discard them otherwise.
We describe two kinds of agents: \emph{conservative} and
\emph{susceptible}, but other criteria for choosing between
incompatible beliefs are possible.
We use the parameter $\agentype$ in the framework to denote the kind of agent.
Conservative agents reject any new belief that contradicts their current set of
beliefs, while susceptible agents always accept new beliefs that replace old
believes if necessary to guarantee a consistent set of beliefs.
Here we present a belief propagation algorithm which describes how agents behave
when faced with a new belief.


Consider a trace $\vtrace$ with $\timestampset_\vtrace = \{\vts_0, \ldots,
\vts_{n-1}, \vts_n\}$.
We use the following notation $\vekb{\vtrace[\vts_j, \vts_k]}{\vagi} =
\bigcup_{\vts \in \{t_j, \ldots, \vts_k\}}
\vekb{\vtrace[\vts]}{\vagi}$.
Also, we introduce the event $\enter(\Believe{\vts}{\vagi} \phi)$ meaning that
\emph{belief $\phi$ enters $i$'s knowledge base at time $\vts$}.
The moment at which this event occurs identifies the moment when a belief is
inserted in an agent's knowledge base, which is crucial to propagate beliefs.
Given a belief $\Believe{\vts_n}{\vagi} \phi$ that is about to enter
$\vekb{\vtrace[\vts_n]}{\vagi}$, i.e., $\snv_{t_{n-1}}
\xrightarrow{\mathsf{enter}(\Believe{\vts_n}{\vagi} \phi),\vts_n} \snv_{t_n}$,
Algorithm~\ref{alg:belief-propagation} propagates the accumulated set of beliefs
as long as they are inside the window $\window$, and resolves conflicts
according to $\agentype$.

Lines 2-3 of Algorithm~\ref{alg:belief-propagation} construct a set $\Psi$ of
candidate beliefs to be propagated---according to \window---together with the
new belief that tries to enter $i$'s \ekb.
The if block (lines 4-6) sorts $\Psi$ according to $\agentype$.
In the foreach block (lines 7-11), we iterate over the sorted list of beliefs
and add them to $\vekb{\vtrace[\vts_n]}{\vagi}$ if they are consistent with the
rest of knowledge and beliefs.
It is easy to see that traversing beliefs from newest to oldest gives
preference to newer beliefs in entering
$\vekb{\vtrace[\vts_n]}{\vagi}$.
In particular, $\Believe{\vts_n}{\vagi} \phi$---the newest
belief---will always enter the $\vekb{\vtrace[\vts_n]}{\vagi}$ unless
this belief contradicts actual knowledge, which corresponds to
susceptible agents.
On the contrary, when sorting from oldest to newest, the older beliefs
will have preference to enter $\vekb{\vtrace[\vts_n]}{\vagi}$, thus,
preventing new inconsistent beliefs to enter
$\vekb{\vtrace[\vts_n]}{\vagi}$, as required for conservative agents.
In particular, $\Believe{\vts_n}{\vagi} \phi$ will not be added to
$\vekb{\vtrace[\vts_n]}{\vagi}$ unless it is consistent with all the previous
beliefs and knowledge.
Finally, we always include the predicate
$\occur{\vts_n}(\enter(\Believe{\vts_n}{\vagi} \phi))$ (line 12) so that the
agent remembers that she was told $\Believe{\vts_n}{\vagi}
\phi$ ---independently on whether she started to believe it.
Note that consistency of $\vekb{\vtrace[\vts_n]}{\vagi}$ in both cases is
directly guaranteed by the inclusion condition in line 8.

\algnewcommand\algorithmicswitch{\textbf{switch}}
\algnewcommand\algorithmiccase{\textbf{case}}
\algnewcommand\algorithmicforeach{\textbf{foreach}}

\algdef{SE}[SWITCH]{Switch}{EndSwitch}[1]{\algorithmicswitch\ #1\ \algorithmicdo}{\algorithmicend\ \algorithmicswitch}%
\algdef{SE}[CASE]{Case}{EndCase}[1]{\algorithmiccase\ #1}{\algorithmicend\ \algorithmiccase}%
\algdef{SE}[FOREACH]{ForEach}{EndForEach}[1]{\algorithmicforeach\ #1\ \algorithmicdo}{\algorithmicend\ \algorithmicforeach}%

\algtext*{EndCase} 

\newcommand{\BeliefAlgWithSwitch}{
\begin{algorithm}
  \caption{Belief propagation}\label{alg:belief-propagation}
  \begin{algorithmic}[1]
    \Procedure{Belief-Propagation}{$\vekb{\vtrace[\vts_n]}{\vagi}$, $\Believe{\vts_n}{\vagi} \phi$, \window, \agentype}
      \State $\Psi \gets \{\Know{\vts_n}{\vagi} \Believe{\vts_n}{\vagi} \psi | \occur{t}(\enter(\Believe{\vts}{\vagi} \phi)) \in \vekb{\vtrace[t_n - \window, t_n]}{\vagi}$ where $ t \in [\vts_n - \window, \vts_n]\}$
      \State $\Psi \gets \Psi \cup \{\Know{\vts_n}{\vagi}\Believe{\vts_n}{\vagi} \phi\}$
      \Switch{$\agentype$}
        \Case{$\susceptible$}
          \State $[b_0, b_1, \ldots, b_n] \gets \mathit{sortNewestOldest}(\Psi)$ 
        \EndCase
        \Case{$\conservative$}
          \State $[b_0, b_1, \ldots, b_n] \gets \mathit{sortOldestNewest}(\Psi)$ 
        \EndCase
      \EndSwitch
      \ForEach{$b$ \textbf{in} $[b_0, b_1, \ldots, b_n]$}
        \If{$\vekb{\vtrace[\vts_0, \vts_{n}]}{\vagi} \cup \{b\} ~ \cancel{\der} ~ \Believe{\vts_n}{\vagi} \bot $}
          \State $\vekb{\vtrace[\vts_n]}{\vagi} \gets \vekb{\vtrace[\vts_n]}{\vagi} \cup \{b\}$
        \EndIf
      \EndForEach
      \State $\vekb{\vtrace[\vts_n]}{\vagi} \gets \vekb{\vtrace[\vts_n]}{\vagi} \cup \{\occur{\vts_n}(\enter(\Believe{\vts_n}{\vagi} \phi))\}$
      \State \textbf{return} $\vekb{\vtrace[\vts_n]}{\vagi}$
    \EndProcedure
  \end{algorithmic}
\end{algorithm}
}

\newcommand{\BeliefAlgWithIf}{
\begin{algorithm}
  \caption{Belief propagation}\label{alg:belief-propagation}
  \begin{algorithmic}[1]
    \Procedure{Belief-Propagation}{$\vekb{\vtrace[\vts_n]}{\vagi}$, $\Believe{\vts_n}{\vagi} \phi$, \window, \agentype}
      \State $\Psi \gets \{\Know{\vts_n}{\vagi} \Believe{\vts_n}{\vagi} \psi | \occur{t}(\enter(\Believe{\vts}{\vagi} \phi)) \in \vekb{\vtrace[t_n - \window, t_n]}{\vagi}$ where $ t \in [\vts_n - \window, \vts_n]\}$
      \State $\Psi \gets \Psi \cup \{\Know{\vts_n}{\vagi}\Believe{\vts_n}{\vagi} \phi\}$
      \If{\hspace{1.83em}$\agentype=\susceptible$\hspace{0.7em}}  $[b_0, b_1, \ldots, b_n] \gets \mathit{sortNewestOldest}(\Psi)$ 
        \ElsIf{$\agentype=\conservative$}
             $[b_0, b_1, \ldots, b_n] \gets \mathit{sortOldestNewest}(\Psi)$ 
        \EndIf
      \ForEach{$b$ \textbf{in} $[b_0, b_1, \ldots, b_n]$}
        \If{$\vekb{\vtrace[\vts_0, \vts_{n}]}{\vagi} \cup \{b\} ~ \cancel{\der} ~ \Believe{\vts_n}{\vagi} \bot $}
          \State $\vekb{\vtrace[\vts_n]}{\vagi} \gets \vekb{\vtrace[\vts_n]}{\vagi} \cup \{b\}$
        \EndIf
      \EndForEach
      \State $\vekb{\vtrace[\vts_n]}{\vagi} \gets \vekb{\vtrace[\vts_n]}{\vagi} \cup \{\occur{\vts_n}(\enter(\Believe{\vts_n}{\vagi} \phi))\}$
      \State \textbf{return} $\vekb{\vtrace[\vts_n]}{\vagi}$
    \EndProcedure
  \end{algorithmic}
\end{algorithm}
}

\BeliefAlgWithIf

\begin{example}
  At $\eightpm$ Alice sends a message to Bob indicating that she is at
  work, so $\vekb{\vtrace[\eightpm]}{\bob}$ contains
  $\occur{\eightpm}(\enter(\Believe{\eightpm}{\bob} \loc^{\eightpm}(\alice, \work)))$
and also
  $\Know{\eightpm}{\bob} \Believe{\eightpm}{\bob}  \loc^{\eightpm}(\alice, \work)$.
%
  At $\tenpm$ Bob checks his Facebook timeline, and he sees a post of
  Charlie---who is a coworker of Alice---at $\eightpm$ saying that he
  is with all his coworkers in a pub having a beer.  Assuming that at
  $\tenpm$ Bob still remembers his belief from $\eightpm$ this new
  information creates a conflict with Bob's beliefs.  Note that
  information from Charlie's post is also taken as a belief since
  there is no way for Bob to validate it.
If Bob is a conservative agent, then
%
$\vekb{\vtrace[\tenpm]}{\bob} =  \{\Know{\tenpm}{\bob} \Believe{\tenpm}{\bob}
\loc^{\eightpm}(\alice, \work)\} ~ \cup \quad
\{\occur{\tenpm}(\enter(\Believe{\tenpm}{\bob} \loc^{\eightpm}(\alice,
\pub)))\},$ meaning that the new belief is rejected.
  If Bob is a susceptible agent:
$\vekb{\vtrace[\tenpm]}{\bob} = \{\Know{\tenpm}{\bob} \Believe{\tenpm}{\bob} \loc^{\eightpm}(\alice, \pub)\}$ $\cup\ \{\occur{\tenpm}(\enter(\Believe{\tenpm}{\bob} \loc^{\eightpm}(\alice, \pub)))\}.$
Bob believes that Alice's location at time $t$ ($\eightpm \leq \vts < \tenpm$)
is work---due to belief propagation.
After $\tenpm$, this belief does not propagate to avoid contradictions with the
new belief $\Believe{\tenpm}{\bob} \loc^{\eightpm}(\alice, \pub)$.\qed
  \label{ex:belief-propagation}
\end{example}

\subsubsection{Semantics of \rtkbl{} (RTKBL)}
\label{sec:semantics-of-rtkbl:rt}

The semantics of \rtkbl\ formulae is given by the following satisfaction relation $\vDash$.

\begin{definition}[Satisfaction Relation]
\label{def:satisfaction-relation:rt}
   Given
   a well-formed trace $ \vtrace \in \traceset $,
   agents $ \vagi, \vagj \in \ag $,
   a finite set of agents $ \vgr \subseteq \ag $,
   formulae $ \vfa, \vfb \in \frtkbl $,
   $ m \in \conSet $,
   $ n \in \actSet $,
   $ o \in \domSet $,
   a variable $ x $,
   an event $ \vevt \in \evtSet $,
   and a time-stamp \vts{},
   the \emph{satisfaction relation $ \sat{} \subseteq \traceset{} \times
   \frtkbl{}$} is defined as shown in Fig.
   \ref{fig:satisfaction-relation:rt}.
\end{definition}

\begin{figure}[!t]
\begin{center}
  \begin{tabular}{lcl}
    $ \vtrace \sat \occur{\vts}(\vevt) $ &
    iff &
    $ (\vrtsn, \vevts, \vts) \in \vtrace $ such that $ \vevt \in \vevts $ \\
    $ \vtrace \sat \neg \vfa $ &
    iff &
    $ \vtrace \; \cancel{\sat} \; \vfa $ \\
    $ \vtrace \sat \vfa \land \vfb $ &
    iff &
    $ \vtrace \sat \vfa $ and $ \vtrace \sat \vfb $ \\
    $ \vtrace \sat \forall \vts \cdot \vfa $ &
    iff &
    for all $ v \in \timestampset_\vtrace $, $ \vtrace \sat \vfa[v / \vts]
    $ \\
    $ \vtrace \sat \forall x : D^{\vts} \cdot \vfa $ &
    iff &
    for all $ v \in D_o^{\vtrace[\vts]} $, $ \vtrace \sat \vfa[v / x]
    $ \\
    $ \vtrace \sat \connectionPredicate{\vts}{m}(\vagi, \vagj) $ &
    iff &
    $ (\vagi, \vagj) \in C_m^{\vtrace[\vts]} $ \\
    $ \vtrace \sat \actionPredicate{\vts}{n}(\vagi, \vagj) $ &
    iff &
    $ (\vagi, \vagj) \in A_n^{\vtrace[\vts]} $ \\
    $ \vtrace \sat \predicate{\vts}{\vterm} $ &
    iff &
    $ \predicate{\vts}{\vterm} \in \kb^{\vtrace[\vts]}_e$\\
    $ \vtrace \sat \Know{\vts}{\vagi} \vfa $ &
    iff &
    $ \bigcup_{\{\vts' \mid \vts' < \vts, \vts' \in \timestampset_\vtrace\}} \kb^{\vtrace[\vts']}_\vag \der (\vfa, \window) $ \\
    $ \vtrace \sat \Believe{\vts}{\vagi} \vfa $ &
    iff &
    $ \bigcup_{\{\vts' \mid \vts' < \vts, \vts' \in \timestampset_\vtrace\}} \kb^{\vtrace[\vts']}_\vag \der (\Believe{\vts}{\vagi} \vfa, \window) $ \\
  \end{tabular}
\end{center}
\caption{Satisfaction relation for \rtkbl}
\label{fig:satisfaction-relation:rt}
\end{figure}

Predicates of type $\occur{t}(e)$ are true if the event $e$ is part of the events that occurred at time $t$ in the trace.
$\forall \vts$ quantifies over all the time-stamps in the trace $\timestampset_\vtrace$, which, as mentioned earlier, is a finite set.
For the remaining domains, $\forall x : D^\vts$, the substitution is carried out over the elements of the domain at a concrete time $t$.
Remember that each individual domian $D^\vts$ always contains a finite set of elements.
However, the same domain at different points in time, e.g., $D^\vts$ and $D^{\vts'}$, for any $\vts \not = \vts'$ might contain different number of elements.
When checking connections $\connectionPredicate{\vts}{m}(i,j)$ and actions $\actionPredicate{\vts}{n}(i,j)$ at time \vts{}, we check whether the corresponding relation--- $C_m^{\vtrace[\vts]}$ and $A_n^{\vtrace[\vts]}$, correspondingly---of the \rtsnm\ at time \vts\ contains the pair of users in question.
Checking whether a predicate of type $\predicate{\vts}{\vv{s}}$ holds is equivalent to looking into the knowledge base of the environment at time \vts.
The environment's knowledge base contains all predicates that are true in the real world at a given moment in time.
For example, ``it is raining in Gothenburg at 19:00'' $\mathit{rain}^{19:00}(\mathit{Gothenburg})$ or ``Alice's location at 20:00 is Madrid'' $\loc^{20:00}(\alice,\mathit{Madrid})$.
Determining whether an agent \vagi\ knows \vf\ at time \vts{}, $\Know{\vts}{\vagi} \vfa$, translates to checking whether \vf~can be derived from \vagi's knowledge base at time \vts.
In order to determine whether a user believes $\phi$ at time $\vts$, $\Believe{\vts}{\vagi} \vfa$, we check whether $\Believe{\vts}{\vagi} \vfa$ is derivable from the knowledge base of the user at time $t$ given the parameter $\window$ of the framework.
This way of defining belief is based on the fact that agents are aware of their beliefs, recall axiom (L2) in Table~\ref{tab:ekb-axioms}.
Therefore if a user believes $\phi$, i.e., $\Believe{\vts}{\vagi} \vfa$, then she must also know that she believes it $\Know{\vts}{\vagi} \Believe{\vts}{\vagi} \vfa$.
Intuitively, $\kb^{\vtrace[\vts]}_\vag \der \phi$ means that ``user $\vag$ knows $\phi$ at time $\vts$'', given its equivalence to the knowledge operator.
As expected, knowledge derivations are also limited by the parameter $\window$.

The intuition behind the previous definitions is better illustrated in the
following example.

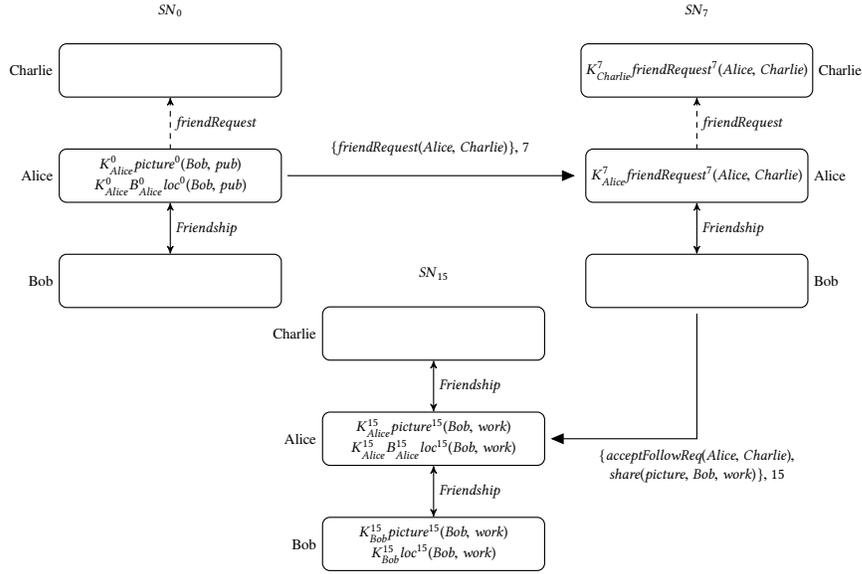
\begin{figure*}[!t]
  \centering
  \scalebox{0.7}{
  \footnotesize
  \begin{tikzpicture}[>=stealth',every text node part/.style={align=center},on grid,auto, minimum height=8ex,semithick]

    \node (alice0) [rectangle, rounded corners, draw, minimum width=120pt, label=left:Alice]{$\Know{0}{\alice} \pic^{0}(\bob, \pub)$ \\
                                                                                             $\Know{0}{\alice} \Believe{0}{\alice} \loc^{0}(\bob, \pub)$};
    \node (bob0) [rectangle, rounded corners, draw,minimum width=120pt,  label=left:Bob][below = 2cm of alice0]{};
    \node (charlie0) [rectangle, rounded corners, draw,minimum width=120pt, label=left:Charlie][above = 2cm of alice0]{};

    \node (sn_0) [fit= (alice0) (bob0) (charlie0), label=above:$\snv_0$] {};

    \node (alice1) [rectangle, rounded corners, draw, minimum width=120pt, label=right:Alice][right = 10cm of alice0]{$\Know{7}{\alice} \friendRequest^{7}(\alice,\charlie)$};
    \node (bob1) [rectangle, rounded corners, draw,minimum width=120pt,label=right:Bob][below = 2cm of alice1]{};
    \node (charlie1) [rectangle, rounded corners, draw,minimum width=120pt,label=right:Charlie][above = 2cm of alice1]{$\Know{7}{\charlie} \friendRequest^{7}(\alice,\charlie)$};

    \node (sn_1) [fit= (alice1) (bob1) (charlie1), label=above:$\snv_7$] {};

    \node (alice2) [rectangle, rounded corners, draw, minimum width=120pt, label=left:Alice][below = 5cm of alice1, xshift=-5cm]{$\Know{15}{\alice} \pic^{15}(\bob, \work)$ \\
    $\Know{15}{\alice} \Believe{15}{\alice} \loc^{15}(\bob, \work)$};
    \node (bob2) [rectangle, rounded corners, draw,minimum width=120pt,label=left:Bob][below = 2cm of alice2]{$\Know{15}{\bob} \pic^{15}(\bob, \work)$ \\
                                                                                                                          $\Know{15}{\bob} \loc^{15}(\bob, \work)$};
    \node (charlie2) [rectangle, rounded corners, draw,minimum width=120pt,label=left:Charlie][above = 2cm of alice2]{};

    \node (sn_2) [fit= (alice2) (bob2) (charlie2), label=above:$\snv_{15}$] {};

    \path[<->] [] (alice0) edge node [] {$\Friendship$} (bob0);
    \path[<-] [dashed] (charlie0) edge node [] {$\friendRequest$} (alice0);

    \path[<->] [] (alice1) edge node [] {$\Friendship$} (bob1);
    \path[<-] [dashed] (charlie1) edge node [] {$\friendRequest$} (alice1);

    \path[<->] [] (alice2) edge node [] {$\Friendship$} (bob2);
    \path[<->] [] (charlie2) edge node [] {$\Friendship$} (alice2);

    \draw[->] [-triangle 45] (sn_0) |- node [xshift=2.7cm]{$\{\friendRequest(\alice,\charlie)\}, 7$} (sn_1);
    \draw[->] [-triangle 45] (sn_1) |- node {$\{\acceptFollowReq(\alice,\charlie),$\\$ \share(\pic, \bob, \work)\}, 15$} (sn_2);
  \end{tikzpicture}}
  \caption{\label{fig:snapchat-example} Example of a Snapchat trace}
\end{figure*}

\begin{example}[Snapchat]
  In this example we model the OSN Snapchat.
  In Snapchat there are two main events that users can perform:
  \begin{inparaenum}[i)]
    \item Connect through a friend relation;
    \item share timed messages (which last up to 10 seconds and can include text and/or a picture)  with their friends.
  \end{inparaenum}
  Fig.~\ref{fig:snapchat-example} shows an example trace for Snapchat.
  The trace consists of:
  A common set of agents $\ag = \{\alice, \bob, \charlie\}$.
  Since \ag\ does not change we avoid using the superindex indicating the time-stamp of the domain.
  Three \snms\ $\snv_0$, $\snv_7$ and $\snv_{15}$.
  The subindex of the \snms\ indicates their time-stamp.

  At time $0$, Alice and Bob are friends, i.e., $\friendship^{0}(\alice,\bob)$.
  This is represented by including the pair $(\alice, \bob)$ in the relation $\Friendship^{\vtrace[0]}$, drawn in the picture as an arrow between Alice and Bob in $\snv_0$.
  This relation between Alice and Bob does not change in \vtrace.
  Moreover, in $\snv_0$, Alice can send a friend request to Charlie.
  It is depicted as an outgoing dashed arrow from Alice's node to Charlie's.
  Thus, $\vtrace \sat \Permitted{\charlie}{\alice}\friendRequest^{0}$ holds.
  Moreover, Alice knows that there is a picture of Bob at the pub, $\pic^{0}(\bob,\pub)$.
  On the other hand, she \emph{believes} that his location is the pub, $\loc^{0}(\bob,\pub)$.
  The reason for this is because she cannot verify that the picture has not been modified or she cannot precisely identify the location.
  However, the existence of $\pic^{0}(\bob,\pub)$ can be verified since it is a picture that Alice can see in the OSN.

  At time $7$, Alice sends a friend request to Charlie.
  Though in this paper we do not discuss modelling the behaviour of events, we assume for this example that Alice can perform this events since he is explicitly permitted.\footnote{We refer the reader to~\cite{Pardo2017} for the definition of operational semantics rules for \ppf\ which describe the behaviour of the events in the OSN.}
  After the execution of the event both agents know $\friendRequest^{7}(\alice,\charlie)$.
  Note that this event produces knowledge.
  This is because the agents can verify that the friend request has occurred.

  Lastly, at time $15$, Charlie accepts Alice's request and Bob shares a picture at work.
  Note that these two events are independent.
  If they were to be executed sequentially, independently of their order, the final \snm\ would always be equal to $\snv_{15}$.
  After Bob's accepting Alice's request $(\alice, \charlie) \not \in \FriendRequest^{\vtrace[15]}$, and $(\alice, \charlie) \in \Friendship^{\vtrace[15]}$.
  That is, Alice cannot send more friend requests to Charlie, and now they have become friends.
  Furthermore, both, Alice and Bob know that Bob shared a picture at work.
  Note that, in this case, Bob also knows that his location is work.
  Nevertheless, Alice believes it.\footnote{For readability we omit $\occur{15}(\enter(\Believe{15}{\alice}\loc^{15}(\bob,\work)))$ in Fig.~\ref{fig:snapchat-example} which is included in $\vekb{\vtrace[15]}{\alice}$.}
  The reason is that, unlikely Bob, Alice cannot confirm that Bob's location is work.

  We mentioned that Snapchat messages last for up to 10 seconds.
  Let us assume w.l.o.g.~that all messages last 10 seconds, i.e., $\window = 10$.
  Given that, in \vtrace, Alice remembers Bob picture from 0 to 10.
  That is,
  $$ \vtrace \sat \forall \vts \cdot 0 \leq \vts \leq 10 \implies \Know{\vts}{\alice} \pic^{0}(\bob,\pub)$$
  Similarly, her belief about Bob location, $\pic^{0}(\bob,\pub)$, vanishes at time 10.
  Note also that, when Charlie accepts Alice's friend request, he still knows (or remembers) that Alice sent it.
  In Snapchat friend request are permanent, but in \rtppf\ we can choose whether friend request disappear after a few seconds.
  It can be done by requiring that the agent knows that a friend request occurred in order to accept it.
  In such a case, in \vtrace, after time $18$ Charlie would not be able to accept Alice's request.

  As mentioned earlier, the purpose of \rtppf\ main goal is expressing privacy policies.
  They can be expressed as a formulae in \rtkbl.
  Let $\ag$ and $\mathit{Locs}$ be domains of agents and locations, respectively.
  For the purpose of this example we assume that they do not change and that is why we do not specify the superindex.
  Alice, who likes keep her weekends private can write the following policy
  $\forall \vagi \colon \ag \cdot \forall l \colon \mathit{Locs} \cdot \forall \vts \cdot \weekend(\vts) \implies (\forall \vts' \cdot \vts' > \vts \implies \Know{\vts'}{\vagi} \pic^{\vts}(\alice,l)).$
  In English it means ``Nobody can know Alice location during the weekend''.
\end{example}


\subsection{Model Checking \rtkbl}
In this section, we show that the model checking problem for \rtkbl\
is decidable.

\begin{theorem}
  The model checking problem for \rtkbl\ is decidable.
  \label{thm:model-checking}
\end{theorem}

\begin{proof}
  Given a trace $\vtrace \in \wftraceset$, a formula $\phi \in \frtkbl$ and a
  window $\vwindow$.
  We show that determining $\vtrace \sat \phi$ is decidable by providing a naive
  model checking algorithm which implements the semantics of \rtkbl, adapting
  the proof for the simpler logic \tkbl\ from~\cite{PKSS16sepposn}.

  We first expand the universal quantifiers in $\phi$ by inductively
  transforming each subformula $\forall x \colon D \cdot \phi'$ and $\forall
  \vts \cdot \phi'$ into a conjunction with one conjunct $\phi'[v/x]$ for each
  element $v$ in the domain $D$ or $\timestampset_\vtrace$.
  All domains---including $\timestampset_\vtrace$---are finite.
  The resulting formula is quantifier free and has size $O(|\phi|\times d^q)$
  where $d$ is a bound on the size of the domain and $q$ is the maximum nested
  stack of quantifiers.
  Let $\phi_1,\ldots,\phi_m$ be the subformulas of the resulting formula,
  ordered respecting the subformula relation.
  An easy induction on $k<m$ shows that we can label every agent and at every
  step of the trace with either $\phi_k$ or $\neg \phi_k$.
  The labelling proceeds from the earliest time-stamp on.
  We begin with the atomic part:
  \begin{compactitem}
    \item
      Checking $\connectionPredicate{t}{e,f}$ and $\actionPredicate{t}{e,f}$ can
      be performed in constant time, simply by checking the model at the given
      instant $t$, for every agent.

    \item
      Checking $\predicate{t}(\vv{s})$ at a given instant $t$ requires one query
      to the epistemic reasoning engine for $\vekb{\vtrace[t]}{e}$ (for the
      environment agent $e$ and time stamp $t$).

    \item
      Checking $\occur{\vts}(e)$ can be performed in constant time by checking
      the set of events $e \in E$ of $\vtrace$ at time $\vts$.

  \end{compactitem}
  Then, for the epistemic part we first resolve all operators:
  \begin{compactitem}
  \item Checking $\psi_k=\neg \psi_j$ and $\psi_k=\psi_j \land \psi_i$
    can be done in constant time for each instant $t$ and agent $i$,
    using the induction hypothesis.

  \item
    First, we construct a set \Axioms\ where we instantiate all the axioms in
    Table~\ref{tab:ekb-axioms} for each $t \in \timestampset_\vtrace$. The
    resulting set has size $|\Axioms| = |\timestampset_\vtrace| \times 11$
    (number of axioms in Table~\ref{tab:ekb-axioms}).
    Secondly, we instantiate KR1 (cf. Table~\ref{tab:ekb-axioms}), for
    $\vwindow$ and for all $t, t' \in \timestampset_\vtrace$ such that $t > t'$
    and $t - t' < \vwindow$.
    The resulting set of axioms has size $O(\sum^{|\timestampset_\vtrace| -
    1}_{n = 1} n \times \vwindow)$.
    That is, all legal combinations of timestamps ($n$) times the window size
    ($\vwindow$).
    These axioms are also included in \Axioms, which, consequently, contains a
    finite set of axioms.
    Finally, checking $K^t_i \psi_j$ and $B^t_i \psi_j$ require one query $ dkd$
    to the epistemic engine for $\Axioms,
    \bigcup_{\{t'|t'<t\in\timestampset_\vtrace\}}\vekb{\vtrace[t]}{i} \der
    \psi_j$.
    The previous query is equivalent to model checking a Kripke structure where
    relations are labelled with triples $(i,t,\vwindow)$. Solving this problem
    is known to be in PSPACE~\cite{FHM+95rk}.

  \end{compactitem}
  It is easy to see that the semantics of \rtkbl{} is captured by this
  algorithm.
  %
\end{proof}
%

\subsection{Properties of the framework}
In this section we present interesting properties of our framework, together with a set of novel derived operators not present in traditional epistemic logics.

\paragraph{To Learn or not to learn; To believe or not to believe}
In \tkbl\ we introduced the learning modality $L_i \phi$ \cite{PKSS16sepposn}.
It stands for $i$ learnt $\phi$ at a moment $t$ when the formula is being evaluated.
Here $L_i \phi$, or more precisely, $\Learnt{t}{i} \phi$ becomes a derived operator.

We say that an agent has learnt $\phi$ at time $\vts$, if she knows $\phi$ at time $\vts$ but she did not know it for any previous timestamp.
Formally, we define it in terms of $\Know{t}{i} \phi$ as follows
$$ \Learnt{t}{i} \phi \triangleq \neg \Know{\pre(t)}{i} \phi \wedge \Know{t}{i} \phi.$$
We can also model when users start to believe something, or \emph{accept} a belief.
To model this concept we define the acceptance operator.
As before, it can be expressed using $\Believe{\vts}{\vagi}$
$$ \LearnBelieve{t}{i} \phi \triangleq \neg \Believe{\pre(t)}{i} \phi \wedge \Believe{t}{i} \phi.$$

Analogously we can express when users \emph{forget} some knowledge or when they \emph{reject} a belief.
Intuitively, an agent forgets $\phi$ at time $\vts$ if she knew it in the previous timestamp, i.e., $\pre(\vts)$---recall the definition of \pre\ in Section~\ref{sec:evolving-osns:rt}---and in \vts\ she does not know $\phi$, and, analogously, for reject.
Formally,
$$ \Forget{t}{i} \phi \triangleq \Know{\pre(t)}{i} \phi \wedge \neg \Know{t}{i} \phi $$
$$ \ForgetBelieve{t}{i} \phi \triangleq \Believe{\pre(t)}{i} \phi \wedge \neg \Believe{t}{i} \phi. $$


\paragraph{Temporal modalities}
The traditional temporal modalities \al\ and \ev\ can easily be defined using quantification over timestamps as follows:
$$ \al \phi(t) \triangleq \forall t \cdot \phi(t) $$
$$ \ev \phi(t) \triangleq \exists t \cdot \phi(t) $$
where $\phi(t)$ is a formula $\phi$ which depends on $t$.

\paragraph{How long do agents remember?}
Agents remember according to the length of the parameter $\window$.
It can be seen as the size of their memory.
Intuitively, increasing agents memory could only increase her knowledge.
We prove this property in the following lemma.
\begin{lemma}[Increasing window and Knowledge]
  Given \vtrace,
  $\vts \in \vtrace$ and
  $\vwindow, \vwindow' \in \nat$ where $\vwindow \leq \vwindow'$.

  If $\vekb{\vtrace[\vts]}{\vagi} \der (\Know{\vts}{\vagi} \phi, \vwindow)$,
  then $\vekb{\vtrace[\vts]}{\vagi} \der (\Know{\vts}{\vagi} \phi, \vwindow')$.
  \label{lemma:increase-window}
\end{lemma}

\begin{proof}
  Assume $\vekb{\vtrace[\vts]}{\vagi} \der (\Know{\vts}{\vagi} \phi, \vwindow)$.
  By Definition~\ref{def:derivation}, there exists a derivation
  $(\phi_1, \vwindow_1)~(\phi_2, \vwindow_2) \ldots (\phi_n, \vwindow_n)$ $\text{ for } n \in \nat$
  such that $(\phi_n, \vwindow_n) = (\phi, \vwindow)$.
  Let $\alpha = \vwindow' - \vwindow$,
  since $\vwindow \leq \vwindow'$ it follows that $\alpha \geq 0$.

  Consider now the following derivation where the same deduction rules as in the previous derivation has been applied, and $\alpha$ is added to each $\vwindow_i$,
  $(\phi_1, \vwindow_1 + \alpha)~(\phi_2, \vwindow_2 + \alpha) \ldots (\phi_n, \vwindow_n + \alpha)$.
  We show now that, if $(\phi_1, \vwindow_1)~(\phi_2, \vwindow_2)$ using a deduction rule $R$ then $(\phi_1, \vwindow_1 + \alpha)~(\phi_2, \vwindow_2 + \alpha)$ can also be derived using $R$, for all $R$ in Table~\ref{tab:kb-deduction}. We split the proof in derivation rules which, copy, reduce or introduce $\vwindow$.
  \begin{itemize}
    \item Rules that copy $\vwindow$. These are, A2, A3, A4, A5, K, B4, B5, L1 and L2. If $\vwindow \in \nat$ given that $\alpha \geq 0$  it trivially follows that $\vwindow + \alpha \in \nat$ which complies with the conditions of any of these rules. In this case $\vwindow = \vwindow'$ therefore the same applies to $\vwindow'$.

    \item Rules that reduce $\vwindow$. This is, KR1. In this case $\vwindow < \vwindow'$. In order for the derivation to be correct both $\vwindow$ and $\vwindow'$ are in $\nat$. Since $\alpha \in \nat$, it follows that $\vwindow + \alpha$ and $\vwindow' + \alpha$ are in $\nat$. Moreover, since $\alpha$ is a constant it also follows that $\vwindow - \vwindow' = (\vwindow + \alpha) - (\vwindow' + \alpha)$. From the previous statement we conclude that the same window increase is required and, therefore the same derivation is performed.

    \item Rules that introduce $\vwindow$. These are, A1, D and Premise. No conditions are imposed in the value of $\vwindow$ in order to apply these rules. Therefore, if they can be applied with window $\vwindow$ since $\alpha \geq 0$ they can also be applied with window $\vwindow + \alpha$.

  \end{itemize}
\end{proof}

We can characterise how long agents remember information depending on the parameters $\window$ and $\agentype$ of the framework.
We differentiate for how long agents remember knowledge or beliefs since the parameter $\agentype$ might influence it.

\begin{lemma}[$\window$ knowledge monotonicity]
  Given \vtrace\ and $\vts \in \timestampset_\vtrace$.
  If $\Know{\vts}{\vagi} \phi \in \vekb{\vtrace[\vts]}{\vagi}$ then
  for all $\vts' \in \timestampset_\vtrace$ such that $\vts \leq \vts' \leq \vts + \window$
  it holds $\vtrace \sat \Know{\vts'}{\vagi} \phi$.\qed
\end{lemma}

\begin{proof}
  Assume $\Know{\vts}{\vagi} \phi \in \vekb{\vtrace[\vts]}{\vagi}$.
  By premise we can derive $(\Know{\vts}{\vagi} \phi, 0)$.
  Let $\vts' = \vts + \window$, by applying KR1 we can derive $(\Know{\vts'}{\vagi} \phi, \window)$.
  By $\sat$ we conclude $\vtrace \sat \Know{\vts'}{\vagi} \phi$.
  Given the above and by Lemma~\ref{lemma:increase-window}, for all $\vts \leq \vts'' < \vts + \window$ it always possible to derive
  $(\Know{\vts''}{\vagi} \phi, \window)$.
\end{proof}

This parameter give us a lot of flexibility in modelling agents memories. By
choosing $\window = \infty$ we can model agents with \emph{perfect recall}, i.e,
agents that never forget. On the other hand, we can also use $\window = 0$,
i.e., agents who do not remember anything.

When $\agentype = \conservative$, memories about beliefs behave in the
same way as knowledge.
Similarly monotonicity results can be proven for beliefs as the lemmas
above.
For example, if $\agentype=\conservative$ then beliefs are preserved until
either forgotten---due to $\window$---or to contradictory knowledge.
Similarly, if $\agentype=\susceptible$, susceptible agents can reject
a belief when exposed to new contradictory beliefs.
Therefore, the duration of their beliefs can be limited by an event
introducing new beliefs in the \ekbs.
Other versions of $\agentype$ are possible, for example based on the
reputation of the agent that emits the information.
It is also possible to consider different $\window$ for different
pieces of information.

\section{Writing Privacy Policies}
\label{sec:rtppl}

In this section we provide a language for writing privacy polices, \rtppl.
In a nutshell \rtppl\ is a restricted version of \rtkbl\ wrapped with \rtpolicy{$~$}{\vagi}{s} to indicate the owner of the policy $\vagi$ and its starting time $s$.

\begin{definition}[Syntax of \rtppl{}]
\label{def:syntax-of-rtppl:rt}
   Given
   agents \( \vaga, \vagb \in \ag \),
   a nonempty set of agents \( G \subseteq \ag \),
   timestamps \( \vstart, \vts \),
   a variable \( x \),
   predicate symbols \( \connectionPredicate{\vts}{m}(\vaga, \vagb), \actionPredicate{\vts}{n}(\vaga, \vagb) \),
   \( \predicate{\vts}(\vterm) \),
   and a formula \( \vfa \in \frtkbl \),
   the \emph{syntax of the real-time privacy policy language} \rtppl{} is inductively defined as:

   \begin{center}
   \begin{tabular}{lcl}
      \( \vpolicy \) & \( ::= \) & \(
         \vpolicy \land \vpolicy \mid
         \forall x. \vpolicy \mid
         \pol{\neg \alpha}{\vag}^\vstart \mid
         \pol{\vfa \then \neg \alpha}{\vag}^\vstart
      \) \\
      \( \alpha \) & \( ::= \) & \(
         \alpha \land \alpha \mid
         \forall x. \alpha \mid
         \exists x. \alpha \mid
         \vfb \mid
         \vfc'
      \) \\
      \( \vfb \) & \( ::= \) & \(
         \connectionPredicate{\vts}{m}(\vaga, \vagb) \mid
         \actionPredicate{\vts}{n}(\vaga, \vagb) \mid
         \occur{\vts}(e)
      \) \\
      \( \vfc' \) & \( ::= \) & \(
         \Know{\vts}{\vag} \vfc \mid
         \Believe{\vts}{\vag} \vfc
      \) \\
      \( \vfc \) & \( ::= \) & \(
         \vfc \land \vfc \mid
         \neg \vfc \mid
         \predicate{\vts}(\vterm) \mid
         \vfc' \mid
         \vfb \mid
         \forall x. \vfc
      \) \\
   \end{tabular}
   \end{center}
\end{definition}

We will use \frtppl{} to denote the set of all privacy policies created according to the previous definition.
To determine whether a policy is violated in an evolving social network, we formalise the notion of conformance for \rtppl.

\begin{definition}[Conformance Relation]
\label{def:conformance-relation:rt}
   Given
   a well-formed trace \( \vtrace \in \traceset \),
   a variable \( x \),
   a timestamp \( \vstart \),
   and an agent \( \vag \in \ag \),
   the \emph{conformance relation \conf{}} is defined as shown in Fig.
   \ref{fig:conformance-relation:rt}.
\end{definition}

\begin{figure}
\begin{center}
\begin{tabular}{lcl}
   \( \vtrace \conf \forall x. \vpolicy \) &
   iff &
   for all \( v \in D_o \), \( \vtrace \conf \vpolicy[v / x] \) \\
   \( \vtrace \conf \pol{\neg \alpha}{\vag}^\vstart \) &
   iff &
   \( \vtrace \sat \neg \alpha \) \\
   \( \vtrace \conf \pol{\vfa \then \neg \alpha}{\vag}^\vstart \) &
   iff &
   \( \vtrace \sat \vfa \then \neg \alpha \) \\
\end{tabular}
\end{center}
\caption{Conformance relation for \rtppl}
\label{fig:conformance-relation:rt}
\end{figure}

The definition is quite simple, especially compared to that of conformance of \tppl{} \cite{PKSS16sepposn}.
If the policy is quantified, we substitute in the usual way.
The main body of the policy in double brackets is dealt with by simply delegating to the satisfaction relation.

\subsubsection{Examples}
\label{sec:rtppl-examples}


\begin{example}
   Assume Alice decides to hide all her weekend locations from her supervisor
   Bob. She has a number of options how to achieve this, depending on what the
   precise meaning of the policy should be.

   If the idea she has is to restrict Bob learning her weekend location
   directly when she posts it, she can define
   \[
      \vpolicy_1 = \forall \vts \cdot \pol{
         \mword{weekend}(\vts) \then \neg K_\vbob^\vts \mword{location}(\valice, \vts)
      }{\valice}^{\mts{2016-04-16}}
   \]
   where the \mword{weekend} predicate is true if the timestamp supplied
   represents a time during a weekend. This policy can be read as ``if \( x \)
   is a time instant during a weekend, then Bob is not allowed to learn at \( x
   \) Alice's location from time \( x \)''.

   This, however, is a very specialized scenario that captures only a small
   number of situations. Bob is, for example, free to learn Alice's location at
   any point not during the weekend, or at any point during the weekend when
   Alice's location is no longer up-to-date.  Though there might be scenarios
   where this might be the desired behavior, we can define a policy that seems
   much closer to the intuitive meaning of learning someone's location on a
   weekend. Consider
   \[
      \vpolicy_2 = \forall \vts \cdot  \pol{
         \mword{weekend}(\vts) \then \neg \exists \vts' \cdot (K_\vbob^{\vts'}
         \mword{location(\valice, \vts))}
      }{\valice}^{\mts{2016-04-16}}.
   \]
   Here, Bob is not allowed to learn Alice's location from a weekend, no matter
   when. If the policy does not get violated, then Alice's weekend locations
   will be completely safe from Bob -- on the social network, at least.
\end{example}

\begin{example}
   One of the advantages of \rtppl{} is the ability to distinguish between the
   original time of a piece of information and the time when it should be
   hidden. Suppose Diane activates the following policy:
   \[
      \vpolicy_1 = \forall {\vts} \cdot \forall x \colon \ag^{\vts} \cdot \pol{
         \neg \mword{friends}^{\vts}(\vdiane, x) \then \neg
         \exists \vts' \cdot (K_x^{\vts} \mword{post}^{\vts'}(\vdiane))
      }{\vdiane}^{\mts{2016-05-28}}
   \]

   This aims to prevent anyone who is not a friend of Diane's from learning any of her posts (here we assume that the \mword{friends} connection is not reflexive for simplicity, otherwise the restriction would target Diane herself, too).

   Though \( \vpolicy_1 \) may seem reasonable enough, it might be unnecessarily
   restrictive. Let us say there is another user, Ethan. Diane becomes friends
   with Ethan on May 31, so when her policy is already in effect. Should Ethan
   be able to learn about Diane's posts from when they were not friends? Not
   according to \( \vpolicy_1 \), which says that no one, regardless of their
   relationship with Diane at the moment, is able to learn about her posts from
   when they were not friends.

   Note that while this may indeed be the desired behaviour, it is, for example,
   not what happens on Facebook, where when two users become friends, they are
   free to access each other's timeline including past events and posts.
   \rtppl{} is expressive enough to model this behaviour as well. We can define:

   \begin{align*}
      \vpolicy_2 = \forall \vts \cdot \forall \vts' \cdot \forall x \colon \ag^{\vts} \pol{
         \neg \mword{friends}^{\vts}(\vdiane, x) \land
         \neg \mword{friends}^{\vts'}(\vdiane, x) \then
         \neg K_x^{\vts'} \mword{post}^{\vts}(\vdiane)
      }{\vdiane}^{\mts{2016-05-28}}
   \end{align*}

   This policy precisely defines the point in time \emph{from} when to hide
   information, \( y \), as well as the point in time \emph{when} to hide it, \(
   y' \). It says, ``if Diane is not friend with someone, then that someone
   cannot learn her posts, but only if they come from a time when they were not
   friends''. Note that \( \vpolicy_2 \) says nothing about users who are
   currently friends of Diane's, which is different from \( \vpolicy_1 \) --
   here her friends can learn anything, including past posts from when they were
   not friends with her.
\end{example}

Since checking conformance of \rtppl\ privacy policies reduces to model checking
the given trace the following corollary follows directly from
Theorem~\ref{thm:model-checking}.
\begin{corollary}
  Checking conformance of \rtppl\ policies is decidable.
\end{corollary}

\section{Related work}
\label{sec:related}

Specifying and reasoning about temporal properties in multi-agent systems using epistemic logic have been previously studied in a number of papers (e.g., in \cite{FHM+95rk} for {\em interpreted systems}; see also \cite{PKSS16sepposn} and references therein).
More recently, Moses \etal~have extended interpreted systems to enhance reasoning about past and future knowledge. In~\cite{ben2013agent} they extend $K_i$ with a time-stamp $K_{i,t}$, allowing for the expression of properties such as ``Alice knows at time 10 that Bob knew $p$ at time 1'', i.e., $K_{\alice,10}K_{\bob,1}~p$.
Thought there are some similarities between the work by Moses \etal~and ours, namely the use of time-stamps in the knowledge modality and the {\em moving} along a trace to place the evaluation in the ``right'' place, there are quite a few differences. First, we differ in the intended use of the logics: Moses \etal~use time to model delays in protocols, whereas our main motivation is to provide a rich privacy policy language for OSNs. Besides, our logic includes belief and other operators  not present in the timed versions introduced by them, and we have time-stamps associated with propositions. We claim, however, that \rtkbl~as at least as expressive as the logics introduced by Moses \etal, though this would need to be formally proved and for that we would need to relate our (non-standard) semantics with interpreted systems. This is left as future work.

Actions have also been taken into account in reasoning about believes and time.
For example, ~\cite{ZDDT15arbat} presents a logic for reasoning about actions
and time.
The logic includes a belief modality, actions and time-stamps for atoms,
modalities and actions.
In our work we do not focus on reasoning about action and time but on
defining dynamic privacy policies for OSNs.
In particular, our policies do not include reasoning about actions.
Also, \cite{ZDDT15arbat} cannot reason about knowledge.
Recently, Xiong~\etal\cite{XASZ17tlt} presented a logic to reason about belief
propagation in Twitter.
The logic includes an (untimed) belief modality and actions, which are used in a
dynamic logic fashion.
Their models are similar to our untimed
SNMs~\cite{PS14fpp,PardoBalliuSchneider2017}.
Even though we do not include actions, we use time-stamps and
knowledge modalities. 
Also, one of the main contributions of our paper is solving
inconsistent beliefs.

As mentioned in the introduction, our work solves the open issues and limitations described in our previous work \cite{PKSS16sepposn}. In that paper we introduced \tfppf, a temporal epistemic framework for describing policies for OSNs. \tfppf~relies on the temporal epistemic logic \tkbl~that allows to express temporal constraints using the classical {\em box} and {\em diamond} temporal operators. Neither the policy language nor the underlying logic \tkbl~have a notion of time: this is only used at the semantic level, allowing to move along a social network model trace in order to interpret the temporal operators. \rtppf~strictly extends \tfppf, so our work not only addressed the already identified limitations identified in \cite{PKSS16sepposn} but also extends that work by allowing the definition of more modalities (e.g., {\em forget}, {\em accept}, {\em belief}) allowing for a more expressive policy language and underlying logic. These allow us to define policies for more complex OSNs, like Snapchat.

Besides the above, it is worth mentioning the work by Wo\'zna \& Lomuscio \cite{WL04lkcrt} where TCTLKD s presented. TCTLKD is a combination of epistemic logic, CTL, a deontic modality and real time.
It is difficult to compare our logic \rtkbl~with TCTLKD as they use CTL, while we have time-stamps in the propositions. The models used to interpret formulae in TCTLKD are based on a semantics for a branching logic, being a combination of timed automata and interpreted systems plus an equivalence relation for modelling permission. Ours is based on (timed) social network models. Besides we can also reason about belief.


%
%

\section{Conclusions}
\label{sec:conclusion}

In this paper, we have presented a novel privacy policy framework
based on a logic that offers explicitly support for expressing
timestamps in events and epistemic operators.
This framework extends~\cite{PardoLic,PS14fpp}, which did not offer
any support for time, and~\cite{PKSS16sepposn} which only had limited
support due to the implicit treatment of time.
Our framework is based on Extended Knowledge Bases.
A query to an EKB starts by instantiating a number of epistemic axioms that
handle knowledge, belief and time (the concrete axioms depend on the OSN
instantiation).
The deductive proof system give an algorithm to deduce the
knowledge of agents acquired at each instant, and in turn a model checking
algorithm for the logic and a check for privacy policy violations.
The explicit time-stamps allow to define learning and forget operators
that capture when knowledge is acquired.
Similarly, one can derive accept and reject operators that model when beliefs
come into existence and are rejected.

The flexibility of the EKBs allows to model different kinds of OSNs in
terms of how the actions affect knowledge and how this knowledge is
preserved through time.
For instance, we can define eternal OSNs like Facebook and ephemeral OSNs like
Snapchat.

Two important avenues for future research are the following. First,
many instantiations enable efficient implementations of checking
privacy policy violations by exploiting whether events can affect the
knowledge of the agents involved. Once the effect of the actions is
fixed one can prove that a distributed algorithm guarantees the same
outcome as the centralized algorithm proposed here. For example, tweets
can only affect the knowledge of subscribers so all other users are
unaffected. Second, once an effective system to check policy
violations is in place, there are different possibilities that the OSN
can offer. One is to enforce the policy by forbidding the action that the
last agent executed, which would lead to the violation. Another can be
the analysis of the trace to assign blame (and correspondingly,
reputation) to the agents involved in the chain of actions. For
example, the creator of a gossip or fake news may be held more
responsible than users forwarding them. Even a finer analysis of
controllability can give more powerful algorithms by detecting which
agents could have prevented the information flow that lead to the
violation. Finally, yet another possibility would be to remove past
events from the history trace of the OSN creating a pruned trace with
no violation. All these possibilities are enabled by having a formal
framework like the one presented in this paper.



\bibliographystyle{ACM-Reference-Format}
\bibliography{references}

%
%

\end{document}